\documentclass[11pt]{article}
\usepackage{amsfonts, amsmath, amssymb, amsthm}
\usepackage{graphicx}
\usepackage[margin=1in]{geometry}
\usepackage{natbib}
\usepackage{setspace}
\usepackage{enumerate}
\usepackage[toc,title,titletoc,header]{appendix}
\usepackage[colorlinks,citecolor=blue,urlcolor=blue, linkcolor=blue]{hyperref}
\usepackage{etoolbox} 
\usepackage{bbm}
\usepackage{xcolor}

\usepackage{multirow}
\def\qed{\rule{2mm}{2mm}}

\usepackage[pagewise,mathlines]{lineno}
\mathchardef\dash="2D

\newtheorem{theorem}{Theorem}

\theoremstyle{definition}

\newtheorem{remark}{Remark}[section]

\AtEndEnvironment{remark}{~\qed}
\AtEndEnvironment{example}{~\qed}

 \renewcommand{\SS}{\mathcal{I}}
 \newcommand{\CalNeg}{\SS_\cal^-}
 \newcommand{\CalPos}{\SS_\cal^+}
 \newcommand{\NCalNeg}{N_\cal^-}
 \newcommand{\NCalPos}{N_\cal^+}
 \newcommand{\NMain}{N_\main}
 \newcommand{\SCalNeg}{S_\cal^-}
 \newcommand{\SCalPos}{S_\cal^+}
 \newcommand{\SMain}{S_\main}

 \newcommand{\Main}{\SS_\main}
 
 \newcommand{\obs}{\mathrm{obs}}
 \renewcommand{\cal}{\mathrm{c}}

 \newcommand{\main}{\mathrm{m}}

 \newcommand{\scaln}{s_{\cal, \obs}^-}
  \newcommand{\scalp}{s_{\cal, \obs}^+}
 \newcommand{\smain}{s_{\main, \obs}}

 \newcommand{\bS}{\boldsymbol{S}}
 \newcommand{\bs}{\boldsymbol{s}}

 \newcommand{\sobs}{\bs_\obs}

 \newcommand{\btheta}{\boldsymbol{\theta}}

 \newcommand{\St}{\mathbb{S}}
\newcommand{\stobs}{\sobs}
\newcommand{\wTheta}{\widehat\Theta}
\newcommand{\wThetaa}{\widehat\Theta_{1-\alpha}}
\newcommand{\wThetac}{\widehat\Theta_{0.95}}

 \newcommand{\alt}{\mathrm{ alt}}
\newcommand{\wThetaalt}{\widehat\Theta_{1-\alpha}^{\alt}}
\newcommand{\wThetacalt}{\widehat\Theta_{0.95}^{\alt}}

\newcommand{\nn}{\nu}
 
\begin{document}

\author{
Panos Toulis\thanks{Email:~\url{panos.toulis@chicagobooth.edu}.
Code is available at~\url{https://github.com/ptoulis/covid-19}.} \\
University of Chicago, \\Booth School of Business \\
}

\bigskip

\title{Estimation of Covid-19 Prevalence from Serology Tests: \\ 
A Partial Identification Approach}

\date{}
\maketitle

\begin{abstract}
We propose a partial identification method for estimating disease prevalence 
from serology studies. Our data are results from antibody tests in some population sample, where the test parameters, such as the true/false positive rates, are unknown. 
Our method scans the entire parameter space, and rejects parameter values using the joint data density as the test statistic.
The proposed method is conservative for marginal inference, in general, but its key advantage over more standard approaches is that it is valid in finite samples even when the underlying model is not point identified. Moreover, our method requires only independence of serology test results, 
and does not rely on asymptotic arguments, normality assumptions, or other approximations.
We use recent Covid-19 serology studies in the US,  and show that the parameter confidence set is generally wide, and cannot support  definite conclusions. 
Specifically, recent serology studies from California suggest a prevalence anywhere in the range 0\%-2\%~(at the time of study),
and are therefore inconclusive. However, this range could be narrowed down 
to 0.7\%-1.5\% if the actual false positive rate of the antibody test was indeed near its empirical estimate~($\sim$0.5\%).
 In another study from New York state, Covid-19 prevalence is confidently estimated in the range 
13\%-17\% in mid-April of 2020,
 which also suggests significant geographic variation in Covid-19 exposure across the US.
 Combining all datasets yields a  5\%-8\% prevalence range.
Our results overall suggest that serology testing on a massive scale can give crucial information for future policy design, even when such tests are imperfect and their parameters unknown.
\end{abstract}

\vspace{20px}

\vspace{-5px}
\noindent 

\thispagestyle{empty} 
\newpage

\setcounter{page}{1}

\section{Introduction} \label{sec:intro}
Since December 2019 the world has been facing the Covid-19 pandemic, 
and its disastrous effects in human life and the economy. 
Responding to the pandemic, most countries have closed off their borders, and imposed unprecedented, universal lockdowns 
on their entire economies. The key reason for such drastic measures is uncertainty: 
we do not yet know the actual transmission rate, the lethality, or the prevalence of this new deadly disease. As governments and policy makers were caught by surprise, there is no doubt that these drastic measures were needed as a first line of defense. The data show that we 
would have to deal with a massive humanitarian disaster otherwise.

At the same time, as the economic pain mounts, especially for the most vulnerable and disadvantaged segments of the population,  there is an urgent need to think of careful ways to safely reopen the economy.
Estimating the true prevalence of Covid-19 has been identified as a key parameter to this effort~\citep{alvarez2020simple}.
In the United States, the number of confirmed Covid-19 cases is 1,193,813 as of May 7 with 
70,802 total deaths. 
This implies a (case) prevalence of 0.36\% (assuming 328m as the US population), 
and a 5.9\% mortality rate of Covid-19, which is even higher than the mortality rate 
reported at times by the World Health Organization.\footnote{The official mortality rate was revised from 2\% in late January to 4\% in early March; see also an official WHO situation report from early March: \url{https://www.who.int/docs/default-source/coronaviruse/situation-reports/20200306-sitrep-46-covid-19.pdf?sfvrsn=96b04adf_2}.}~However, the true prevalence, that is, the number of people who are currently infected
 or have been infected by Covid-19 over the entire population is likely much higher, 
and so the mortality rate should be significantly lower than 5.9\%.
A growing literature is attempting to estimate these numbers through epidemiological models~\citep{li2020substantial, flaxman2020report}, or structural assumptions
~\citep{hortaccsu2020estimating}.
 
A more robust alternative seems to be possible through randomized serological studies that detect marker antibodies indicating exposure to Covid-19. In the US, there is currently a massive coordinated effort to evaluate the widespread application of these tests. The results 
 are expected in late May of 2020.\footnote{CDC page: \url{https://www.cdc.gov/coronavirus/2019-ncov/lab/serology-testing.html}.}
The hope is that these tests will determine the true prevalence of the virus, and thus its lethality, 
and also determine whether someone is immune enough to return to work~(the extent of immunity is still uncertain, however).
Furthermore, seroprevalence studies can give information on risk factors for the disease, such as a patient's age, location, or underlying health conditions. They may also reveal important medical information on immune responses to the virus, such as how long antibodies last in people’s bodies following infection, 
and could also identify those able to donate blood plasma, which is a possible treatment to seriously ill Covid-19 patients.\footnote{Food and Drug Administration (FDA) announcement on serology studies~(04/07/2020):~\url{https://www.fda.gov/news-events/press-announcements/coronavirus-covid-19-update-serological-tests}.}
The development of serology tests is therefore essential to designing a careful strategy towards both  effective medical treatments and a gradual reopening of the  economy.

Until widespread serology testing is possible, however, we have to rely on a limited number of serology studies that have started to emerge in various areas of the globe, including the US. 
Table~\ref{tab1} presents a non-exhaustive summary of such studies around the world.
For example, in Germany, serology tests in early April  showed a 14\% prevalence in a sample of 500 people.\footnote{Report in German:~\url{https://www.land.nrw/sites/default/files/asset/document/zwischenergebnis_covid19_case_study_gangelt_0.pdf}.}
In the Netherlands, a study in mid-April showed 
a lower prevalence at 3.5\% in a small sample of blood donors.\footnote{Presentation slides in Dutch:~\url{https://www.tweedekamer.nl/sites/default/files/atoms/files/tb_jaap_van_dissel_1604_1.pdf}. 
}~\footnote{
See also a summary of these projects in the journal ``Science":~\url{https://www.sciencemag.org/news/2020/04/antibody-surveys-suggesting-vast-undercount-coronavirus-infections-may-be-unreliable}.
}~In the US, in a recent and relatively large study in Santa Clara, California,~\citet{bendavid2020covid} 
estimated an in-sample prevalence of 1.5\% from 50 positive test results in a sample of 3330 patients. 
Using a reweighing technique, the authors extrapolated this estimate to 2-4\% prevalence
in the general population. A follow-up study in LA County found 35 positives out of 846 tests. 
What is unique about these last two studies is that data from a prior validation study are also available, where, say, 401 ``true negatives" were tested with 2 positive results, implying a false positive rate of 0.5\%.
Upon publication, these studies received intense criticism because the false positive rate appears to be large enough compared to the underlying disease prevalence. For example, 
 the Agresti-Coull and Clopper-Pearson 95\% confidence intervals for the false positive rate are $[0.014\%, 1.92\%]$ and $[0.06\%, 1.79\%]$, respectively.
These intervals for the false positive rate are not incompatible even with a 0\% prevalence, since a 1.5\% false positive rate achieves $0.015 \times 3330 \approx 50$ (false) positives on average, same as the observed value in the sample.

Such standard methods, however, are justified based on approximations, asymptotic arguments, prior specifications (for Bayesian methods), or normality assumptions, which are always suspect in small samples.~In this paper, we develop a method that can assess finite-sample statistical significance in a robust way.
The key idea is to treat all unknown quantities as parameters, and explore the entire parameter space to assess agreement with the observed data. 
Our method adopts the partial identification framework, where the goal is not to produce point estimates, but to identify sets of plausible parameter values~\citep{wooldridge2007s, tamer2010partial, chernozhukov2007estimation, manski2003partial, manski2010partial,manski2007partial, romano2008inference, romano2010inference, honore2006bounds, imbens2004confidence, beresteanu2012partial, stoye2009more, kaido2019confidence}.
Within that literature, our proposed method appears to be unique in the sense that it constructs a procedure that is valid in finite samples given the correct distribution of the test statistic. Importantly, the choice of the test statistic can affect only the power of our method, but not its validity. 
Such flexibility may be especially valuable in choosing a test statistic that is both powerful and easy to compute. Thus, the main benefit of our approach is that it is valid with {\em enough computation}, whereas 
more standard methods are only valid with {\em enough samples}.

 The rest of this paper is structured as follows. In Section~\ref{sec:setup} we describe the problem formally. In Section~\ref{sec:method_overview} we describe the proposed method on a high level.
A more detailed analysis along with a modicum of theory is given in Section~\ref{sec:method_details}.
In Section~\ref{sec:application} we apply the proposed method 
on data from the Santa Clara study, the LA County study, and a recent study from New York state. 
 
 \renewcommand{\arraystretch}{1.4}
 \begin{table}[]
\small
\begin{tabular}{llcll}
Prevalence  & Location    & Time  & Method     & Notes                                                                                                                                                        \\
\hline
6.14\%      & China       & 01/21 & PCR      & Sample of 342 pneumonia patients~\citep{liang2020prevalence}.                                                                                                                          \\
2.6\%       & Italy       & 02/21 & PCR      & Used 80\% of entire population in V\'{o}, Italy~\citep{lavezzo2020suppression}.                                                                                                                             \\
5.3\%       & USA         & 03/12 & PCR      & Used 131 patients with ILI symptoms~\citep{spellberg2020community}.                                                                                                        \\
13.7\% & USA & 03/22 & PCR & Sample of 215 pregnant women in NYC~\citep{sutton2020universal}.
\\
0.34\%      & USA         & 03/17 & model    & \citep{yadlowsky2020estimation}.                                                                                                                                                           \\
9.4\%       & Spain       & 03/28 & serology & Sample of 578 healthcare workers~\citep{alberto2020}.  \\
3\% & Japan       & 03/31 & serology &
Random set of 1000 blood samples in Kobe Hospital~\citep{doi2020estimation}.
\\
36\%        & USA         & 04/02 & PCR      & Study in large homeless shelter in Boston~\citep{baggett2020prevalence}.                                                                                                                             \\
1.5\% & USA         & 04/03 & serology & Recruited 3330 people via Facebook~\citep{bendavid2020covid}.   
 \\
2.5\%       & USA         & 04/04 & model    & Uses ILINet data;
implies 96\% unreported cases~\citep{lu2020estimating}.
  \\
  9.1\% & Switzerland & 04/06 & serology & 
  Sample of 1335 individuals in Geneva~\citep{stringhini2020repeated}.
  \\
  14\%        & Germany     & 04/09 & serology & Self-reported 400 households
  in Gangelt.~\href{https://www.tweedekamer.nl/sites/default/files/atoms/files/tb_jaap_van_dissel_1604_1.pdf}{(Streeck et al., 2020).}                                                                                                \\
3.1\% & Netherlands & 04/16 & antigen  & Used 3\% of all blood donors~\href{https://www.land.nrw/sites/default/files/asset/document/zwischenergebnis_covid19_case_study_gangelt_0.pdf}{(Jaap van Dissel, 2020)}.                                                                                                                         
\\
0.4-40\% & USA & 04/24 & model & 
Partial identification using \#cases/deaths~\citep{manski2020estimating}.
\\
\hline
\hline
Unreported       &  & & \\             
\hline
90\%       & USA         & 03/16 & model    & 
Used airflight data to identify transmission rates~\citep{hortaccsu2020estimating}.
                                                                                                                        \\
85\%         & India       & 04/02         & model               & Extrapolated from respiratory-related 
 cases~\citep{venkatesan2020estimate}.                               \\                                                                                        
\end{tabular}
\caption{\footnotesize Summary of recent Covid-19 prevalence studies crudely categorized 
as either statistical models or medical tests (PCR or serology). Most studies report intervals, 
but here we mainly report midpoints.~{\em Top panel:}~PCR stands for ``polymerase chain reaction", and is a key test to detect presence of the virus' RNA; {\em serology} denotes a serology test to detect presence of antibodies (e.g., IgA, IgM, IgG). {\em ILI} stands for ``influenza-like illness" and describes methods that use data recorded 
from patients with general influenza-like symptoms, including but not limited to Covid-19 cases.~{\em Bottom panel}:~Studies that aimed mainly to estimate percentage of Covid-19 cases that are not reported. This number can be used to calculate prevalence estimates, 
but needs to be adjusted for the exact timeframe of the study; e.g., the analysis of 
~\citet{hortaccsu2020estimating} implies a 1.5-2.5\% prevalence in Santa Clara county in mid-March.
}
\label{tab1}
\end{table} 

\section{Problem Setup}\label{sec:setup}
Here, we formalize the statistical problem of estimating disease prevalence through imperfect medical tests. 
Every individual $i$ is associated with a binary status $x_i$: it is $x_i= 1$ if the individual has developed antibodies from exposure to the disease, and $x_i=0$ if not. We will also refer to these cases as ``positive" and ``negative",  respectively. 
Patient status is not directly observed, but can be estimated with a serology (antibody) test.

This medical antibody test can be represented by  a function $t: \{0, 1\} \to \{0, 1\}$, 
and determines whether someone is positive or  negative. As usual, the categorization of the test results can be described through the following table:

\renewcommand{\arraystretch}{1.5}
\begin{table}[h!]
\centering
\begin{tabular}{l  l | lll}
                                                                                            &   & \multicolumn{2}{c}{true status, $x=$}               &  \\
                                                                                            &   & \multicolumn{1}{c}{0} & \multicolumn{1}{c}{1} &  \\
                                                                                            \hline
\multirow{2}{*}{\begin{tabular}[c]{@{}l@{}}$t(x)=$\end{tabular}} & 0 & true negative         & false negative        &  \\
                                                                                           & 1 & false positive        & true positive         & 
\end{tabular}
\end{table}
\noindent We will assume that each test result is an independent random outcome, such that the true positive rate and false positive rate, denoted respectively by $q$ and $p$,\footnote{The terms ``sensitivity" and ``specificity" are frequently used in practice of medical testing. 
In our setting, sensitivity maps to the true positive rate $(q)$, 
and specificity maps to one minus the false positive rate $(1-p)$. 
In this paper, we will only use the terms ``true/false positive rate" as they are more precise and self-explanatory.
}
are constant:
\begin{align}\label{eq:A12}
P\left \{ t(x) =1 | x =1 \right\} = q, \tag{A1} \\
P\left \{ t(x) =1 | x =0 \right\} = p. \tag{A2}
\end{align}
This assumption may be untenable in practice. In general, 
patient characteristics, or test target and delivery conditions can affect the test results. For example, 
\citet{bendavid2020covid} report slightly different test performance characteristics depending on which antibody (either IgM or IgG) was being detected. 
We note, however, that this assumption is not strictly necessary for the validity of our proposed inference procedure. It is only useful in order to obtain a precise calculation for the distribution of the test statistic~(see Theorem~\ref{thm} and remarks).

To determine test performance characteristics, and gain information about the true/false positive rates of the antibody test, there is usually a {\em validation study} where the 
underlying status of participating individuals is known. In the Covid-19 case, for example, such validation study could include pre-Covid-19 blood samples that have been preserved, and are thus ``true negatives".
To simplify, we assume that in the validation study there is a set $\CalNeg$ of participating individuals, where it is known that everyone is a true negative, and a set $\CalPos$ where everyone is positive; 
i.e.,
$$
x_i=0,~\text{for all}~i\in \CalNeg,\text{and}~x_i=1,~\text{for all}~i\in \CalPos.
$$
There is also the {\em main study} with a set $\Main$ of participating individuals, where the true status is not known. We assume no overlap between sets $\CalNeg, \CalPos$ and $\Main$, which is a 
realistic assumption.
We define  $\NCalNeg = |\CalNeg|$ and $\NCalPos=|\CalPos|$ as the respective number of participants in the 
validation 
study, and $N_\main = |\SS_\main|$ as the number of participants in the main study.
These numbers are observed, but the full patient sets or the patient characteristics, may not be observed.

We also observe the positive test results in both studies:
\begin{align}\label{eq:obs}
\SCalNeg & = \sum_{i\in\CalNeg} t(x_i),~\SCalPos = \sum_{i\in\CalPos} t(x_{i}),
\nonumber & \text{[Calibration study]} \\
\SMain & = \sum_{i\in \Main} t(x_i).
& \text{[Main study]}
\end{align}
Thus, $\SCalNeg$ is the number of false positives in the validation study since we know that 
all individuals in $\CalNeg$ are true negatives. Similarly, $\SCalPos$ is the number of true positives in the validation study since all individuals in $\CalPos$ are known to be positive.
These numbers offer some simple estimates of the false positive rate and true positive rate of the medical test, respectively: $\hat p = \SCalNeg / \NCalNeg$ and $\hat q = \SCalPos / \NCalPos$.
We use $(\scaln, \scalp, \smain)$ to denote the observed values of test positives
$(\SCalNeg, \SCalPos, \SMain)$, respectively, which are integer-valued random variables. 

The statistical task is therefore to use observed data 
$\{(\NCalNeg, \NCalPos, N_\main), (\scaln, \scalp, \smain)\}$ and do  inference on the quantity:
\begin{equation}\label{eq:pi}
\quad \pi = \frac{\sum_{i\in \SS_\main} x_i}{N_\main},
\end{equation}
i.e., the unknown disease prevalence in the main study. We emphasize 
that $\pi$ is a finite-population estimand --- we discuss (briefly) the issue of extrapolation to the general population in Section~\ref{sec:discussion}.
The challenge here is that $\SMain$ generally includes both false positives and true positives, 
which depends on the unknown test parameters, namely the true/false positive rates $q$ and $p$.
Since $\pi N_\main$ is the (unknown) number of infected individuals in the main study, 
we can use Assumptions (A1) and (A2) to write down this decomposition formally:
\begin{align}\label{eq:model}
\SCalNeg & \sim \mathrm{Binom}(\NCalNeg, p),~\SCalPos  \sim \mathrm{Binom}(\NCalPos, q),\text{and}~\nonumber\\
S_\main & \sim \underbrace{\mathrm{Binom}(\pi N_\main, q)}_{\text{true positives}} +  
\underbrace{\mathrm{Binom}(N_\main - \pi N_\main, p)}_{\text{false positives}},
\end{align}
where $\mathrm{Binom}$ denotes the binomial random variable.  
For brevity, we define $\bS = (\SCalNeg, \SCalPos, \SMain)$ as our joint data statistic,
and $\btheta = (p, q, \pi)$ as the joint parameter value.
The independence of tests implies that the density of $\bS$ can be computed exactly as follows.
\begin{equation}\label{eq:f}
f(\bS \mid \btheta) = \mathrm{d}(\SCalNeg; \NCalNeg, p) \cdot
\mathrm{d}(\SCalPos; \NCalPos, q) \cdot
\left[\sum_{j=0}^{j=\SMain}  \mathrm{d}(j; \pi \NMain, q) \cdot \mathrm{d}(\SMain-j; \NMain-\pi \NMain, p)\right],
\end{equation}
where $\mathrm{d}(k; n, s)$ denotes the probability of $k$ successes in a binomial 
experiment with $n$ trials and $s$ probability of success. There are several ways to implement Equation~\eqref{eq:f} efficiently --- we defer discussions on computational issues to Section~\ref{sec:discussion}.

\section{Method}\label{sec:method}
We begin with an illustrative example 
to describe the proposed method on a high level. 
We give more details along with some theoretical guarantees in the section that follows.

\subsection{Illustrative example}\label{sec:method_overview}
Let us consider the Santa Clara study~\citep{bendavid2020covid} with observed data:
$$
(\NCalNeg, \NCalPos, \NMain) =(401, 197, 3330),~\text{and}~ 
(\scaln, \scalp, \smain) = (2, 178, 50).
$$
The unknown quantities in our analysis are $q,p$ and $\pi$: the true positive rate of the test, the false positive rate, and the unknown prevalence in the main study, respectively.
Assume zero prevalence ($\pi=0\%$), 90\% true positive rate ($q=0.9$), and 
1.5\% false positive rate ($p=0.015$). We ask the question:
``Is the combination $(p, q, \pi) = (0.015, 0.90, 0)$ compatible with the data?".
Naturally, this can be framed in statistical terms as a null hypothesis:
\begin{equation}\label{eq:H0}
H_0: (p, q, \pi) = (0.015, 0.90, 0).
\end{equation}
To test $H_0$ we have to compare the observed positive test results with the values that {\em could have been observed} 
if indeed the true model parameter values were $(p, q, \pi) = (0.015, 0.90, 0)$.
Our model is simple enough that we can execute this hypothetical analysis exactly based on the density of $f(\bS | \btheta)$ in Equation~\eqref{eq:f}, where $\bS = (\SCalNeg, \SCalPos, \SMain)$ is the 
vector of all positive test results, and $\btheta$ is specified as in $H_0$;
see Figure~\ref{fig:example1}. 
\begin{figure}[t!]
\centering
\includegraphics[scale=0.5]{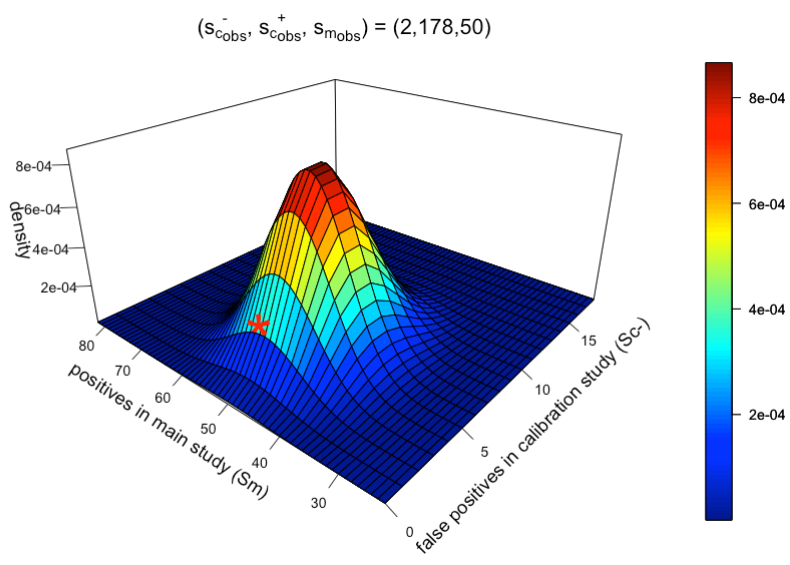}
\caption{Density $f(\bS | H_0)$ of test positives $\bS = (\SCalNeg, \SCalPos, \SMain)$ conditional on $H_0$ of Equation~\eqref{eq:H0}, and with fixed $\SCalPos=178$ (its observed sample value) to ease visualization.
 The observed values $(\scaln, \smain) = (2, 50)$ are marked with an asterisk. To test $H_0$ we need to calculate some kind of ``p-value" for the observed point. In our construction, we simply test whether the density at the observed values exceeds an appropriately chosen threshold~(see Section~\ref{sec:method_details}). }
\label{fig:example1}
\end{figure}
~To simplify visualization, in Figure~\ref{fig:example1} we fix the component $\SCalPos$ of $\bS$ to its observed value 
$(\SCalPos= 178$), and only plot the density with respect to the other two components, $(\SCalNeg, \SMain)$; i.e., we plot $f(\bS \mid H_0, \SCalPos=178)$.
One can visualize the full joint distribution $f(\bS \mid H_0)$ as a collection of such 
conditional densities for all possible values of $\SCalPos$.

The next step is to decide whether the observed value of $\bS$, namely 
$\sobs = (2,178, 50)$, is compatible with the distribution of Figure~\ref{fig:example1}. 
We see that the mode of the distribution is around the point $(\SCalNeg, \SMain) = (5, 45)$, whereas the point $(2, 50)$ is at the lower edge of the distribution. If the observed values were even further, say at $(\SCalNeg, \SMain) = (2, 80)$, then 
we could confidently reject $H_0$ since the density at $(2, 80)$ is basically zero. Here, we have to be careful because the actual observed values are still somewhat likely under $H_0$. 
Our method essentially accepts $H_0$ when the density of this distribution at the observed value $\sobs$ of statistic $\bS$ is above some 
threshold $c_0$, that is, we decide based on the following rule:
\begin{equation}\label{eq:method}
\text{Accept}~H_0~\text{if}~f(\sobs | H_0) > c_0.
\end{equation}

The test in Equation~\eqref{eq:method} is reminiscent of the likelihood ratio test, the key difference being that our test does not require maximizations of the likelihood function over the parameter space, which is computationally intensive, and frequently unstable numerically --- we make a concrete comparison in the application of  Section~\ref{sec:compare2}. Our test essentially uses the density of $\bS$ as the test statistic for $H_0$, 
while threshold $c_0$ generally depends on the particular null values being tested. Assuming that the test of Equation~\eqref{eq:method} has been defined, we can then test for all possible combinations of our parameter values, $\btheta \in \Theta$, in some large enough parameter space $\Theta$, 
and then invert this procedure in order to construct the confidence set.
As usual, we would like this confidence set to cover the true parameters with some minimum probability~(e.g., 95\%). 
In the following section, we show that this is possible through an appropriate construction of the test in Equation~\eqref{eq:method}, which takes into account the level sets of the density function depicted in Figure~\ref{fig:example1}. The overall procedure is computationally intensive, but is valid in finite samples without the need of asymptotic or normality assumptions.
The details of this construction, including the appropriate selection 
of the test threshold and the proof of validity,  are presented in the following section.

\subsection{Theoretical Details}\label{sec:method_details}
Let $\bS = (\SCalNeg, \SCalPos, \SMain) \in \mathbb{S}$ denote the statistic, 
where $\St=\{0, \ldots, \NCalNeg\} \times \{0, \ldots, \NCalPos\} \times \{0, \ldots, \NMain\}$,
and let $\btheta = (p, q, \pi) \in\Theta$ be the model parameters.
We take $\Theta$ to be finite and discrete; e.g., for probabilities 
we take a grid of values between 0 and 1. 
Let $\sobs = (\scaln, \scalp, \smain)$ denote the observed value of $\bS$ in the sample.
Let $f(\bS | \btheta)$ denote the density of the joint statistic conditional on the model parameter value $\btheta$, as defined in Equation~\eqref{eq:f}.
Suppose that $\btheta_0$ is the true unknown parameter value, and assume that 
\begin{equation}
P(\btheta_0 \in \Theta) = 1.\tag{A4}
\end{equation}
Assumption (A4) basically posits that our discretization is fine enough to include the true parameter value with probability one. In our application, this assumption is rather mild as we are dealing with parameters that are either probabilities or integers, and so bounded within well-defined ranges.
Moreover, this assumption is implicit essentially in all empirical work since computers operate with finite precision.
Our goal is to construct a confidence set $\wTheta_{1-\alpha} \subseteq\Theta$ such that
$
P(\btheta_0 \in\wTheta_{1-\alpha}) \ge 1-\alpha,
$
where $\alpha$ is some desired level (e.g., $\alpha=0.05$).
Trivially, $\wTheta_{1-\alpha}=\Theta$ satisfies this criterion, so we will aim to make $\wTheta$ as narrow as possible. We will also need the following definition:
\begin{equation}\label{eq:nu}
\nn(z, \btheta) = \big | \{ \bs\in\mathbb{S} : 0 < f(\bs | \btheta) \le z \} \big |.
\end{equation}
Function $\nu$ depends on level sets of $f$, and counts the number of sample data points~(over the sample space $\St$) with likelihood at $\btheta$ that is smaller than the observed likelihood at $\btheta$.

We can now prove the following theorem.

\newcommand{\zth}{z^*_{\btheta}}
\newcommand{\ztho}{z^*_{\btheta_0}}

\begin{theorem}\label{thm}
Suppose that Assumption (A4) holds. Consider the following construction for the confidence set:
\begin{equation}\label{eq:Theta1}
\wThetaa = \big\{\btheta \in \Theta : f(\sobs | \btheta) \cdot \nn(f(\sobs | \btheta), \btheta) > \alpha \big\}.
\end{equation}
Then, $P(\btheta_0 \in \wThetaa) \ge 1-\alpha$.
\end{theorem}
\begin{proof}
For any fixed $\btheta\in\Theta$ consider the function $g(z, \btheta) = z \cdot \nn(z, \btheta)$, $z\in[0,1]$.
Note that, for fixed $\btheta$, function $g(z, \btheta)$ is monotone increasing and generally not continuous 
with respect to $z$.
Let $\mathbb{F}_{\btheta}= \{ f(\bs | \btheta) : \bs\in\St\}$, and define 
as $\zth$ the unique fixed point for which 
$g(\zth, \btheta)- \alpha = 0$, if that point exists; if not, define 
the point as $\zth = \max\{ z \in \mathbb{F}_{\btheta} : g(z, \btheta) \le \alpha$\}. 
It follows that $g(\zth, \btheta) \le \alpha$, for any $\btheta$, and so the event
$\{ f(\sobs | \btheta) \cdot \nn(f(\sobs | \btheta), \btheta) \le \alpha\}
= \{ g(f(\sobs | \btheta), \btheta) \le \alpha\}$
is the same as the event $\{f(\stobs | \btheta) \le \zth\}$. 
Now we can bound the coverage  probability as follows.

\begin{align}
P(\btheta_0 \notin \wThetaa) 
& = P\left\{ f(\stobs | \btheta_0) \cdot \nn(f(\stobs | \btheta_0), \btheta_0) \le \alpha \right\}\nonumber\\
& = P\left\{ f(\stobs | \btheta_0)  \le \ztho \right\} = \sum_{\bs\in\mathbb{S}} \mathbb{I}\left\{ f(\bs | \btheta_0)  \le 
\ztho \right\} f(\bs | \btheta_0) \\
& \le \ztho \sum_{\bs\in\mathbb{S} : 0 < f(\bs | \btheta_0)\le \ztho} 1  = 
\ztho \cdot \nn(\ztho, \btheta_0) = g(\ztho, \btheta_0) \le  \alpha.\nonumber
\end{align}
In the first line we used the test definition and Assumption~(A4); in the second line, we used 
the monotonicity of $g$, and the fact that $\btheta_0$ is the true parameter value; in the last line, we used the definition of $\zth$ and the uniform bound on $g$.
\end{proof}

When $z^\ast_{\btheta}$ is not a discontinuity point of $g$, for all $\btheta$, then our 
test is exact in the sense that $P(\btheta_0\in\wTheta_{1-\alpha}) = 1-\alpha$.
In general, however, this condition will not hold for all $\Theta$, and so the confidence set
of Equation~\eqref{eq:Theta1} may be conservative and lose power. 
We could potentially achieve more power if instead we define the confidence set as follows:
\begin{equation}\label{eq:Theta2}
\wThetaalt = \left\{\btheta\in\Theta:
\sum_{\bs\in\mathbb{S}} \mathbb{I}\big\{f(\bs | \btheta) \le f(\sobs | \btheta) \big\} f(\bs | \btheta) > \alpha \right\}.
\end{equation}
\begin{theorem}\label{thm2}
Suppose that Assumption (A4) holds. Consider the construction of confidence set
$\wThetaalt$ as defined in Equation~\eqref{eq:Theta2}.
Then, $P(\btheta_0 \in \wThetaalt) \ge 1-\alpha$.
\end{theorem}
\begin{proof}
The proof is almost identical to Theorem~\ref{thm}, if we replace the definition of $g$ with $g(z, \btheta) = \sum_{\bs\in\mathbb{S}} \mathbb{I}\left\{f(\bs | \btheta) \le z \right\} f(\bs | \btheta)$, $z\in[0,1]$. With this definition $g$ is a smoother function,
which explains intuitively why this construction will generally lead to more power.
\end{proof}

It is straightforward to see that $\wThetaalt \subseteq \wThetaa$ almost surely since 
$$
\sum_{\bs\in\mathbb{S}} \mathbb{I}\big\{f(\bs | \btheta) \le f(\sobs | \btheta) \big\} f(\bs | \btheta)
\le f(\sobs | \btheta) \sum_{\bs\in\mathbb{S}} \mathbb{I}\big\{f(\bs | \btheta) \le f(\sobs | \btheta) \big\}
= f(\sobs|\btheta) \nu(f(\sobs|\btheta), \btheta).
$$

Since both constructions are  valid in finite samples, the choice between $\wThetaa$ or $\wThetaalt$ should be mainly based on computational feasibility. 
The construction of $\wThetaa$ may be easier to compute 
in practice as it depends on a summary of 
the distribution $f(\bs | \btheta)$ through the level set function $\nu$, while the 
construction of $\wThetaalt$ requires full  knowledge of the entire distribution. If it is 
computationally feasible, however, $\wThetaalt$ should be preferred because 
it is contained in $\wThetaa$ with probability one, as argued above.
This leads to sharper inference.  See also the applications on serology studies in Section~\ref{sec:application} for more details, where the 
construction of $\wThetaalt$ is feasible.

\subsection{Concrete Procedure and Remarks}\label{sec:proc}
Theorems~\ref{thm} and~\ref{thm2} imply the following simple procedure to construct a 95\% confidence set:
\begin{center}
\large
\textsc{Procedure 1}
\end{center}
\vspace{-10px}
\begin{enumerate}
\item Observe the value $\sobs=(\scaln, \scalp, \smain)$ of test positives in the validation study and the main study. 
Create a grid $\Theta \subset [0, 1]^3$ for the unknown parameters $\btheta=(p, q, \pi)$.
\item For every $\btheta$ in $\Theta$, calculate $f(\sobs | \btheta)$ 
as in Equation~\eqref{eq:f}, and $\nu(\sobs, \btheta)$ as 
in Equation~\eqref{eq:nu}.
\item Reject all values $\btheta \in \Theta$ for which 
$f(\stobs | \btheta) \cdot \nn(f(\sobs | \btheta), \btheta) \le 0.05$; alternatively, we can reject based on 
$\sum_{\bs\in\mathbb{S}} \mathbb{I}\big\{f(\bs | \btheta) \le f(\sobs | \btheta) \big\} f(\bs | \btheta) \le 0.05$.

\item The remaining values in $\Theta$ form the 95\% confidence sets $\wThetac$ or 
$\wThetacalt$, respectively.
\end{enumerate}

\begin{remark}[Computation]
Procedure 1 is fully parallelizable 
over $\btheta$, and so the main computational difficultly is the need to sum over the sample space $\St$.
Note, however, that Procedure 1 can work for any choice of $\bS$ given its density $f(\bS | \btheta)$. 
Thus, our method offers valuable flexibility for inference; for instance, $f$ could be simulated, or $\bS$ 
could be a simple statistic~(e.g., sample averages) and not necessarily an ``$\arg\max$" estimator. See Section~\ref{sec:discussion} for more discussion on computation. 
\end{remark}

\begin{remark}[Identification]\label{remark:id}
Procedure 1 is not a typical partial identification method in the sense that 
there are settings under which the model of Equation~\eqref{eq:f} is point identified~(i.e., when $\NMain,\NCalNeg, \NCalPos\to\infty$).
However, we choose to describe Procedure 1 as a partial identification method for two 
main reasons. First, it is more plausible, in practice, that the calibration studies are small and finite~($\NCalNeg, \NCalPos <\infty$), since a calibration study needs to have 
high-quality, ground-truth data.
Second, it can happen that we don't have both kinds of calibration studies available~(i.e.,
 it could be that either 
$\NCalNeg=0$ or $\NCalPos=0$). 
In both of these settings, the underlying model is no longer point identified, and so Procedure 1 is technically a partial identification method.
\end{remark}

\begin{remark}[Conservativeness]\label{remark:cons}
Procedure 1 generally produces conservative confidence intervals.
However, we can show that $P(\btheta_0\notin\wThetaalt) \ge \alpha - \epsilon$, where 
$\epsilon = \max_{s\in \St} P\left\{ f(\bS | \btheta_0) = f(s| \btheta_0) \right\}$. 
This value is very small, in general~(e.g., $\epsilon \sim 10^{-3}$ in the 
Santa Clara study). 
In this case, the alternative construction is ``approximately exact" in the sense that the coverage probability of $\wThetaalt$ is almost equal to $(1-\alpha)$.
\end{remark}

\begin{remark}[Marginal inference]\label{remark:marginal}
The parameter of interest in our application could only be the disease prevalence, whereas the true/false positive rates of the antibody test may be considered ``nuisance". In this paper, we directly project $\wThetaa$~(or $\wThetaalt$)
on a single dimension to perform marginal inference~(see Section~\ref{sec:application}), but this is generally conservative, especially at the boundary of the parameter space~\citep{stoye2009more, kaido2019confidence, chen2018monte}.
A sharper way to do marginal inference with our procedures is an interesting direction for future work.
\end{remark}

\subsection{Comparison to other methods}\label{sec:compare1}
How does our method compare to a more standard  frequentist or Bayesian approach?
Here, we discuss two key differences.
First, as we have repeatedly emphasized in this paper, our method is valid in finite samples under only independence of test results, which is a mild assumption. In contrast, a standard frequentist approach, say based on the bootstrap, is inherently approximate and relies on asymptotics, while a Bayesian method requires the specification of priors and posterior sampling. 
Of course, our procedure requires more computation, mainly compared to the bootstrap, and can be conservative for marginal inference~(see Remark~\ref{remark:marginal}), but 
this is arguably a small price to pay in a critical application such as the estimation of Covid-19 prevalence.

A second, more subtle, difference is the way our method performs inference. Specifically, we decide whether any $\btheta\in\Theta$ is in the 
confidence set based on the entire density $f(\bs|\btheta)$ over all $\bs\in\St$, whereas
both frequentist and Bayesian methods typically perform inference ``around the mode" of the likelihood function $f(\sobs | \btheta)$ with fixed $\sobs\in\St$~(we ignore how the prior specification affects Bayesian inference to simplify exposition).
This can explain, on an intuitive level, how the inferences of the respective methods may differ. 
Figure~\ref{fig:compare} illustrates the difference. On the left panel, we 
plot the likelihood, $f(\sobs|\btheta)$, as a function of $\btheta\in\Theta$. 
Typically, in frequentist or Bayesian methods, the confidence set is around the mode, say $\hat\btheta$. We see that a parameter value, say $\btheta_1$, with a likelihood value,  $f(\sobs |\btheta_1)$, that is low in absolute terms will generally not be included in the confidence set.
However, in our approach, the 
 value $f(\sobs | \btheta_1)$ is not important in absolute terms for doing 
 inference, but is only important relative to all other values $\{f(\bs | \btheta_1) : \bs\in\St\}$ of the test statistic distribution $f(\bs|\btheta_1)$. 
 Such inference will typically include the mode, $\hat\btheta$, but will also 
 include parameter values at the tails of the likelihood function, such as $\btheta_1$. As such, our method is expected to give more accurate inference in small-sample problems, 
 or in settings with poor identifiability where the likelihood is non-smooth and multimodal. We argue that we actually see these effects  in the application on Covid-19 serology studies analyzed in the following section --- see also Section~\ref{sec:compare2} and Appendix~\ref{appendix:example} for concrete numerical examples.

\begin{figure}[t!]
\centering
\includegraphics[scale=0.25]{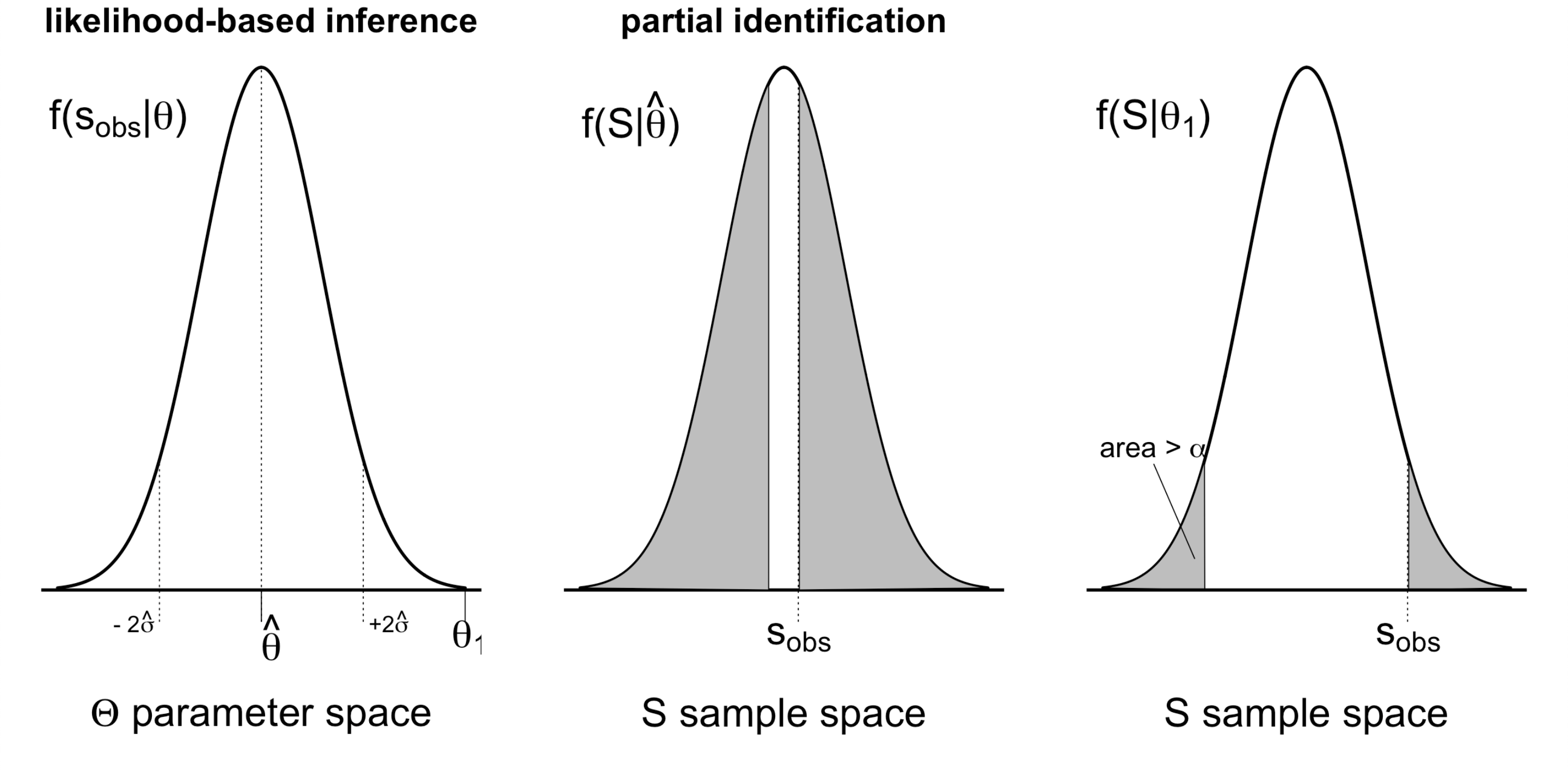}
\caption{\small Illustration of main difference between standard methods of inference, and the partial identification method in this paper.
{\em Left:}~In standard methods, inference is typically based on the likelihood 
$f(\sobs|\btheta)$ as a function over $\Theta$, and around some mode $\hat\btheta$. 
Parameter values with low likelihood, such as $\btheta_1$, 
are not included in the confidence set.~{\em Middle \& right:}~In our method, inference is based on the entire distribution function $f(\bs | \btheta)$ over the statistic parameter space, $\St$. As such, $\hat\btheta$ will usually be in the confidence set~(middle plot). Moreover, $\btheta_1$ will also be in the confidence set 
if $f(\sobs | \btheta_1)$ is high  relative to the rest of $f(\bs | \btheta_1)$ even when 
$f(\sobs|\btheta_1)$ is small relative to the mode $f(\sobs | \hat\btheta)$~(right plot).
}
\label{fig:compare}
\end{figure}

\section{Application}\label{sec:application}
In this section, we apply the inference procedure of Section~\ref{sec:proc} to several
serology test datasets in the US. Moreover, we present results for combinations of these datasets, assuming that the tests have identical specifications.
This is likely an untenable assumption, but it helps to illustrate how we can use our approach to flexibly combine all evidence.  Before we present the analysis, we first discuss some data on serology test performance to inform our inference.

\subsection{Serology test performance}\label{sec:serology}
An important aspect of serology studies are the test performance characteristics. 
As of May 2020, there are perhaps more than a hundred commercial serology tests in the US, but they can differ substantially across manufacturers and technologies.
In our application, we use data from~\citet{bendavid2020covid}, who 
applied a serology testing kit distributed by Premier Biotech. 
\citet{bendavid2020covid} used validation test results provided by the test manufacturer, and also performed 
a local validation study in the lab. The combined validation study estimated a true positive rate of 80.3\% (95\% CI: 72.1\%-87\%), and a false positive rate of 0.5\% (95\% CI: 0.1\%-1.7\%).

To get an idea about how these performance characteristics relate to other available serology tests we 
use a dataset published by the FDA based on benchmarking 12 other testing kits to grant emergency use authorization (EUA) status. 
The dataset is summarized in Table~\ref{tab:sero}.
 We see that the characteristics of the testing kit used by \citet{bendavid2020covid}  are compatible with the FDA data shown in the table.
 For example, a true positive rate of 80\% is below the mean and median of the point estimates in the FDA dataset.
A false positive rate of 0.5\% falls between the median and mean of the respective FDA point estimates.
A reason for this skewness is likely the existence of one outlier testing kit that performs notably worse than the others (Chembio Diagnostic Systems).  Removing this datapoint brings the mean false positive rate down to 0.6\%, very close to the estimate provided by~\citet{bendavid2020covid}.

 \renewcommand{\arraystretch}{1.2}
\begin{table}[h!]
\centering
\begin{tabular}{ll | cccc}
            Test characteristic         &                & min  & median & mean & max  \\
                                                                                   \hline
\multirow{3}{*}{true positive rate ($q$)}   & point estimate & 77.4\% & 91.1\%   & 90.6\% & 100\%  \\
                                                                                   & 95\% CI, low endpoint   & 60.2\% & 81\%     & 80.1\% & 95.8\% \\
                                                                                   & 95\% CI, high endpoint  & 88.5\% & 96.7\%   & 95.3\% & 100\%  \\
                                                                                   \hline
\multirow{3}{*}{false positive rate ($p$)} &     point estimate           & 0\%    & 0.35\%   & 1.07\% & 5.6\% \\
                                                                                   &    95\% CI, low endpoint            & 0\%    & 0.1\%    & 0.52\% & 2.7\% \\
                                                                                   &     95\% CI, high endpoint           & 0.3\%  & 1.6\%   & 3.18\% & 11.1\%
\end{tabular}
\caption{\small Performance characteristics of 12 different testing kits granted with emergency 
authorization status by the FDA. {\em Source:}~Author calculations based on publicly available FDA dataset at \url{https://www.fda.gov/medical-devices/emergency-situations-medical-devices/eua-authorized-serology-test-performance?mod=article_inline}.
}
\label{tab:sero}
\end{table}

\subsection{Santa Clara study}
In the Santa Clara study,~\citet{bendavid2020covid} report a validation study and 
main study, with $(\NCalNeg, \NCalPos, \NMain) = (401, 197, 3330)$ participants, respectively. 
The observed test positives are $\sobs = (\scaln, \scalp, \smain) = (2, 178, 50)$, respectively. Given these data, we produce the 95\% confidence sets  for $(p, q, \pi)$ following both procedures in \eqref{eq:Theta1} and \eqref{eq:Theta2} described in Section~\ref{sec:proc}. 
%
In Figure~\ref{fig:sc_Theta} of Appendix~\ref{appendix:serology}, we jointly plot all triples in the 3-dimensional space $\wTheta_{0.95}$ of  Equation~\eqref{eq:Theta1}, with additional coloring based on prevalence values.
We see that the confidence set is a convex space tilting to higher prevalence values as the false positive rate of the test decreases. The true positive rate does not affect prevalence, as long as it stays in the range 80\%-95\%.

To better visualize the pairwise relationships between the model parameters, we also provide Figure~\ref{fig:sc} that breaks down Figure~\ref{fig:sc_Theta} into two subplots, one visualizing the pairs $(\pi, p)$ and another visualizing the pairs $(\pi, q)$. 
The figure visualizes both $\wThetac$ and $\wThetacalt$ to illustrate the differences between the two 
constructions.
From Figure~\ref{fig:sc}, we see that the Santa Clara study is not conclusive about  Covid-19 prevalence. A prevalence of 0\% is plausible, 
given a high enough false positive rate.
However, if the true false positive rate is near its empirical value of  0.5\%, as estimated by~\citet{bendavid2020covid}, then the identified prevalence rate is estimated in the range 0.4\%-1.8\% in $\wThetac$. Under this assumption, we see that $\wThetacalt$ offers a sharper inference, as expected, with an estimated prevalence in the range 0.7\%-1.5\%.
Even though, strictly speaking, the statistical evidence is not  sufficient here for definite inference on prevalence, we tend to favor the latter interval because (i) common sense precludes 0\% prevalence in the Santa Clara county~(total pop. of about 2 million); (ii) the interval generally agrees with the test performance data presented earlier, and (iii) it is still in the low end compared 
to prevalence estimates from other serology studies~(see Table~\ref{tab1}).
Regardless, pinning down the false positive rate is important for estimating prevalence, especially when prevalence is as low as it appears to be in the Santa Clara study. Roughly speaking, a decrease of 1\% in the false positive rate implies an increase of 1.3\% in prevalence.

\begin{figure}[t!]
\vspace{-30px}
\centering
$\wThetac$ in Santa Clara study
\includegraphics[scale=0.38]{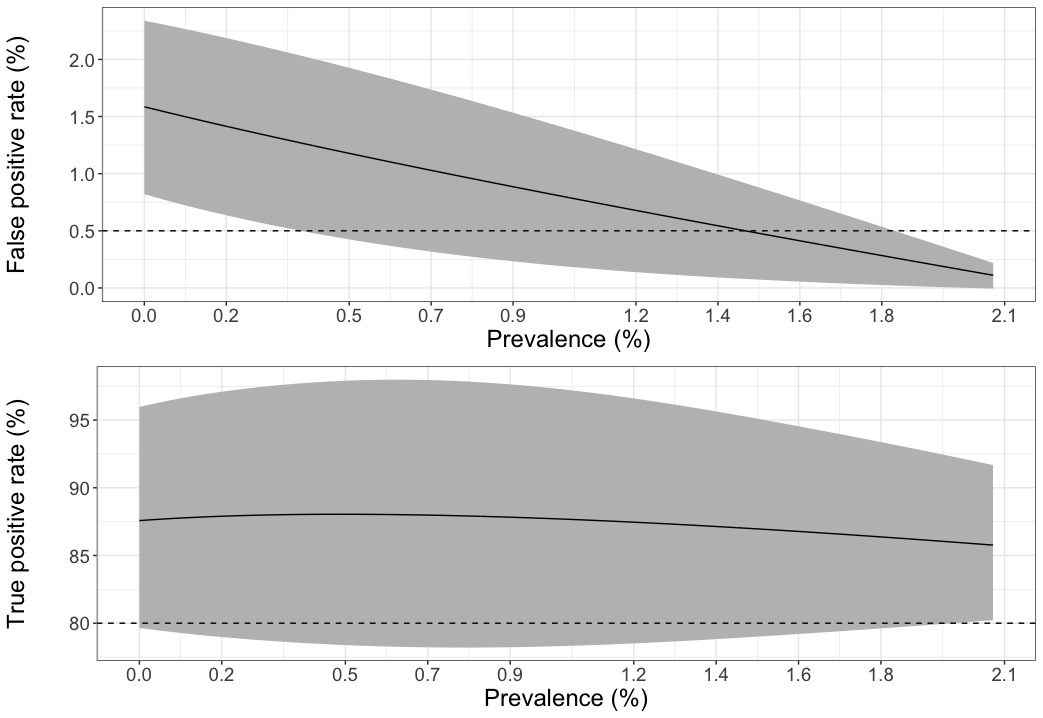}

\vspace{5px}
$\wThetacalt$ in Santa Clara study
\includegraphics[scale=0.38]{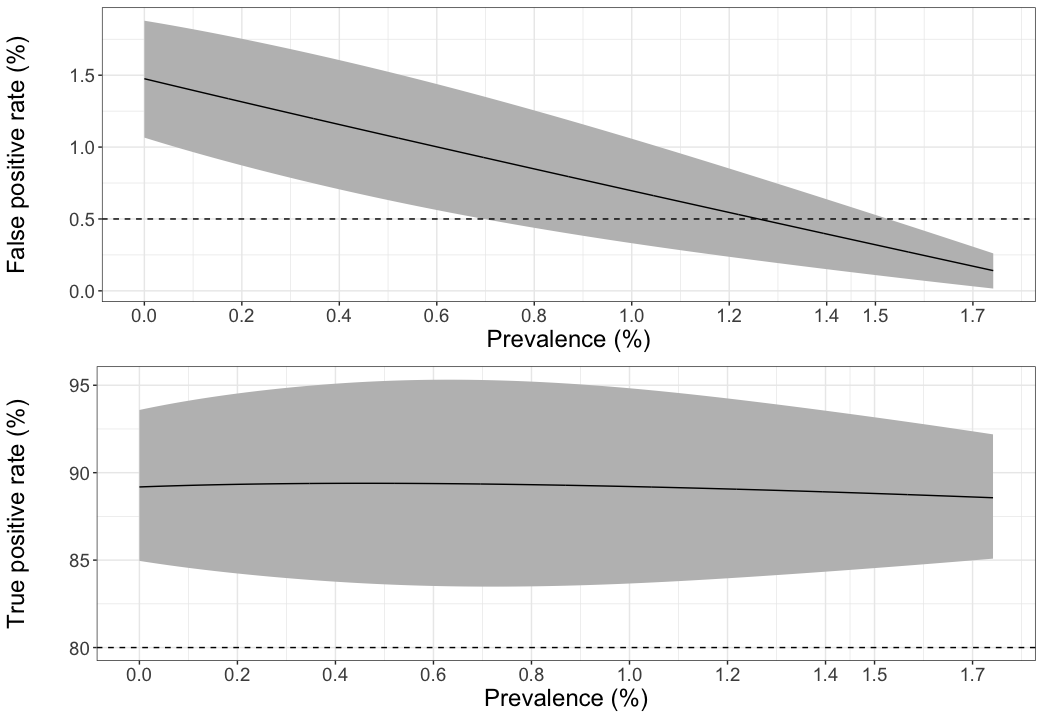}
\caption{
\small
Confidence sets $\wThetac$ and $\wThetacalt$ of the basic and the alternate 95\% confidence set construction in Equations~\eqref{eq:Theta1} and~\eqref{eq:Theta2}, respectively, for the Santa Clara study.  Each confidence set is broken down to pairwise relationships. All values in a shaded area are in the corresponding confidence set. The horizontal dashed line corresponds to the empirical estimate of the false positive rate~($\hat p = 0.5\%$) of the serology test; the dashed line corresponds to the estimate of its true positive rate ($\hat q = 80\%$);~see Section~\ref{sec:serology} for details.
}
\label{fig:sc}
\end{figure}

\subsection{Santa Clara study: Comparison to other methods}
\label{sec:compare2}
In this section, we aim to discuss how our method practically compares to more standard methods using data from the Santa Clara study. In our comparison we include 
Bayesian methods, a classical likelihood ratio-based test, and the Monte Carlo-based approach 
 to partial identification proposed by~\citet{chen2018monte}.\footnote{
All these methods are fully implemented in the accompanying code at \url{https://github.com/ptoulis/covid-19}.
 }

\subsubsection{Comparison to standard Bayesian and bootstrap-based methods}
Due to initial criticism, the authors of the original Santa Clara study published a revision of their work, where they use a bootstrap procedure to calculate 
confidence intervals for prevalence in the range 0.7\%-1.8\%.\footnote{Link: \url{https://www.medrxiv.org/content/10.1101/2020.04.14.20062463v2.full.pdf}.}
Some recent Bayesian analyses report wider prevalence
intervals in the range 0.3\%-2.1\%~\citep{gelman2020bayesian}. In another Bayesian multi-level analysis, \citet{levesque2020note} report similar findings but mention that posterior summarization here may be 
subtle, since the posterior density of prevalence in their specification includes 0\%.
These results are in agreement with our analysis in the previous section only if we assume 
that the true false positive rate of the serology test was near its empirical estimate~($\sim$0.5\%).
We discussed intuitively the reason for such discrepancy in Section~\ref{sec:compare1}, where we argued that standard methods typically do inference ``around the mode" of the likelihood, and may thus miscalculate the amount of statistical information hidden in the tails.


For a numerical illustration, consider two parameter values, namely
$\btheta_1 = (0.5\%, 90\%, 1.2\%)$ and $\btheta_2 = (1.5\%, 80\%, 0\%)$, 
where  the components denote the false positive rate, true positive rate, and prevalence, respectively. In the Santa Clara study, 
$f(\sobs | \btheta_1) = 2.2 \times 10^{-3}$ and $f(\sobs | \btheta_2) = 9.58 \times 10^{-8}$, that is, $\btheta_2$~(which implies 0\% prevalence) maps to a likelihood value that is many orders of magnitude smaller than $\btheta_1$.
In fact, $\btheta_1$ is close to the mode of the likelihood, 
and so frequentist or Bayesian inference is mostly based around that mode, ignoring the tails of the likelihood function, such as $\btheta_2$.
For our method, however, the small value of $f(\sobs | \btheta_2)$ is more-or-less irrelevant --- what matters is how this value compares to the entire 
distribution $f(\bs | \btheta_2)$. It turns out that 
$f(\sobs | \btheta_2) \nu(f(\sobs | \btheta_2), \btheta_2) = 0.137$, that is, 
13.7\% of the mass of $f(\bs|\btheta_2)$ is below the observed level $f(\sobs | \btheta_2) = 9.58 \times 10^{-8}$. As such, $\btheta_2$ cannot be rejected at the 5\% level~(see also Appendix~\ref{appendix:example}).
This highlights the key difference of our procedure compared to frequentist or Bayesian procedures. More generally, we expect to see such important differences between the inference from our method and the inference from other more standard methods in settings with small samples or poor identification~(e.g., non-separable, multimodal likelihood). 

\subsubsection{Comparison to likehood ratio test}\label{sec:LRT}
As briefly described in Section~\ref{sec:method_overview}, our test is related to the likelihood ratio test~\citep[Chapter 3]{lehmann2006testing}. Here, we study the similarities and differences between the two tests, both theoretically and empirically through the Santa Clara study. Specifically, consider testing a null hypothesis that the true parameter is equal to some value $\btheta$  using  the likelihood ratio statistic,
\begin{equation}\label{eq:LRT}
t(\sobs | \btheta) = \frac{f(\sobs | \btheta)}{\max_{\btheta'\neq\btheta} f(\sobs | \btheta')}.
\end{equation}

Since $f$ is known analytically from Equation~\eqref{eq:f}, the null distribution 
of $t(\bS | \btheta)$ can be fully simulated. An exact p-value can then be obtained 
by comparing this null distribution with the observed value $t(\sobs|\btheta)$.
We can see that this method is similar to ours in the sense that both methods 
use the full density $f(\bS|\btheta)$ in the test, and both are exact. The main difference, however, is that our method is using a 
summary of the density values $f(\bS | \btheta)$ that are below the observed value $f(\sobs | \btheta)$, 
which avoids the expensive~(and sometimes numerically unstable) maximization in the denominator of the likelihood ratio test in~\eqref{eq:LRT}. Our proposed method turns out to be orders of magnitude faster 
than the likelihood ratio approach as we get $50$-fold to $200$-fold speedups in our setup --- see Section~\ref{sec:discuss_computation} for a more detailed comparison in computational efficiency.

To efficiently compare the inference between the two tests, we sampled 
5,000 different parameter values from inside $\wThetac$ --- i.e., the 95\% confidence set from the basic 
test in Equation~\eqref{eq:Theta1} --- and 5,000 parameter values from $\Theta\setminus \wThetac$, and then calculated the overlap between the test decisions.
The likelihood ratio test rejected 3\% of the values from the first set, and 
98\% of the values from the second set, indicating a good amount of overlap between the two tests. 
The correlation between the p-value from the likelihood ratio test, and the 
 values $f(\sobs | \btheta) \nu(f(\sobs | \btheta), \btheta)$, which our basic test
uses to make a decision in Equation~\eqref{eq:Theta1}, was equal to 0.94.
The correlation with the alternative confidence set construction is 0.90, 
using instead the values $\mathbb{I}\big\{f(\bs | \btheta) \le f(\sobs | \btheta) \big\} f(\bs | \btheta)$ 
in the above calculation.
Since the likelihood ratio test is exact, these results suggest that our test procedures are generally high-powered.
%

In Figures~\ref{fig:LRT_1} and~\ref{fig:LRT_2}  of Appendix~\ref{appendix:LRT}, we 
plot the 95\% confidence sets from the likelihood ratio test described above for the Santa Clara study and the LA county study~(of the following section). 
The estimated prevalence is 0\%-1.9\% for Santa Clara, which 
is shorter than $\wThetac$ but wider than $\wThetacalt$, as reported earlier; the same holds for LA county. 
As with our method, prevalence here is estimated through direct projection of the confidence set, which may be conservative.
 It is also possible that with more samples the likelihood ratio test could achieve the same interval as $\wThetacalt$~(we used only 100 samples), but this would come at an increased computational cost. 
 Overall, the likelihood ratio test produces very similar results to our method, but it is not as efficient computationally.

\subsubsection{Comparison to Monte Carlo confidence set method 
of~\citet{chen2018monte}}\label{sec:mcmc}
In recent work, ~\citet{chen2018monte} proposed a Monte Carlo-based method of inference 
in partially identified models. The idea is to sample from a quasi-posterior distribution, and then calculate $q_n$, the 95\% percentile of 
$\{ f(\sobs | \btheta^{(j)}),~j=1, \ldots\}$, where $\btheta^{(j)}$ denotes the $j$-th sample from the posterior.
The 95\% confidence set is then defined as: 
\begin{equation}\label{eq:mcmc}
\wTheta = \{\btheta\in\Theta: f(\sobs | \btheta) \ge q_n\}.
\vspace{-10px}
\end{equation}

We implemented this procedure with an MCMC chain that appears to be mixing well --- see Appendix~\ref{appendix:MCMC} and Figure~\ref{fig:mcmc2} for details. 
The 95\% confidence set, $\wTheta$,  is given in Figure~\ref{fig:mcmc} of Appendix~\ref{appendix:MCMC}. 
Simple projection, yields a prevalence in the range 0.9\%-1.43\%. This suggests that our MCMC ``spends more time" around the mode of the likelihood, which we back up 
with numerical evidence in Appendix~\ref{appendix:MCMC}.
Finally, we also tried Procedure 3 of~\citet{chen2018monte}, which does not require MCMC simulations but is generally more conservative. Prevalence was estimated in the range~0.12\%-1.65\%, 
which is comparable to our method and the likelihood ratio test.

\subsection{LA County study}
Next, we analyze the results from a recent serology study in  Los Angeles county, which estimated a prevalence of 4.1\% over the entire county population.\footnote{\url{http://publichealth.lacounty.gov/phcommon/public/media/mediapubhpdetail.cfm?prid=2328}} 
We use the same validation study as before since this study was executed by the same team as the Santa Clara one.
Here, the main study had $\NMain=846$ participants with $\smain=35$ positives.\footnote{
This number was not reported in the official study announcement mentioned above. It was reported in a Science article referencing one of the authors of the study: \url{https://www.sciencemag.org/news/2020/04/antibody-surveys-suggesting-vast-undercount-coronavirus-infections-may-be-unreliable}.}  
For inference, we only use the alternative construction, $\wThetacalt$, 
of Equation~\eqref{eq:Theta2} to simplify exposition. The results are shown in Figure~\ref{fig:LA}.

\begin{figure}[t!]
\centering
\vspace{-10px}
$\wThetacalt$ in LA county study\\

\includegraphics[scale=0.45]{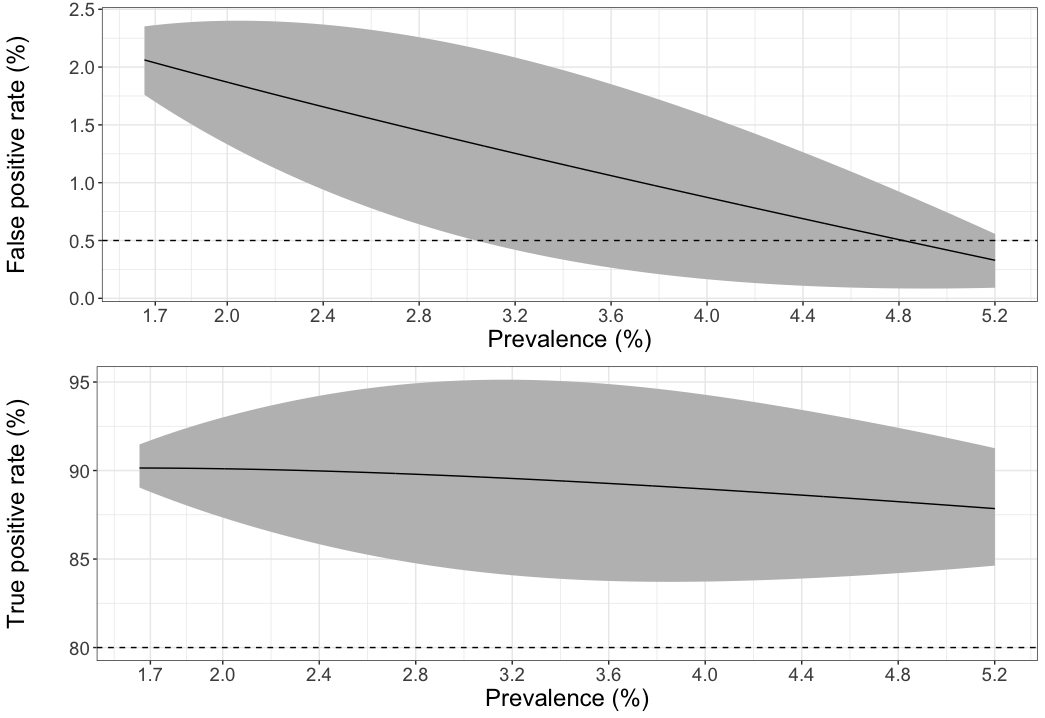}

\caption{ Visualization of $\wThetacalt$ for the LA county study. 
We see that the evidence for Covid-19 prevalence are stronger than the Santa Clara study. 
Prevalence is estimated in the range 1.7\%-5.2\%. 
This is shortened to 3\%-5.2\% if we assume a 0.5\% false positive rate for the antibody test.}
\label{fig:LA}
\end{figure}

In contrast to the Santa Clara study, 
we see that the results from this study are conclusive. 
The prevalence rate is estimated in the range 1.7\%-5.2\%. If the false positive rate is, for example, closer to its empirical estimate~(0.5\%) 
then the identified prevalence is relatively high, somewhere in the range 3\%-5.2\%. 
We also see that the true positive rate is 
estimated in the range 85\%-95\%, which is higher than the empirical point estimate of 80\% provided by~\citet{bendavid2020covid}. 
In fact, the empirical point estimate is not even in the 95\% confidence set.
Finally,  as an illustration, we combine the data from the Santa Clara and LA county studies.
The assumption is that the characteristics of the tests used in both studies were identical.
%
The results are shown in Figure~\ref{fig:SC+LA} of Appendix~\ref{appendix:serology}.
We see that 0\% prevalence is consistent with the combined study as well. 
Furthermore, prevalence values higher than 2.5\% do not seem plausible in the combined data.

\subsection{New York study}
Recently, a quasi-randomized study was conducted in New York state, including NYC, 
which sampled individuals shopping in grocery stores.
Details about this study were not made available. 
Here, we assume that the medical testing technology used was the same
as in the Santa Clara and LA county studies, or at least similar enough that the comparison remains informative. 

%
Under this assumption, we can use the same validation study as before,
with $(\NCalNeg, \NCalPos) = (401, 197)$ participants in the validation study, 
and $(\scaln, \scalp)=(2, 178)$ positives, respectively.
The main study in New York had $\NMain=3000$ participants with $\smain=420$ 
observed test positives.\footnote{
\url{https://www.nytimes.com/2020/04/23/nyregion/coronavirus-antibodies-test-ny.html}}  
The $\wThetacalt$ confidence set on this dataset is shown in Figure~\ref{fig:NY}.
We see that the evidence in this study is much stronger than the Santa Clara/LA county studies 
with an estimated prevalence in the range 12.9\%-16.6\%.
The true positive rate is now an 
 important identifying parameter in the sense that knowing its true value could narrow down the confidence set even further.

\begin{figure}[t!]
\centering
$\wThetacalt$ in New York study\\

\includegraphics[scale=0.45]{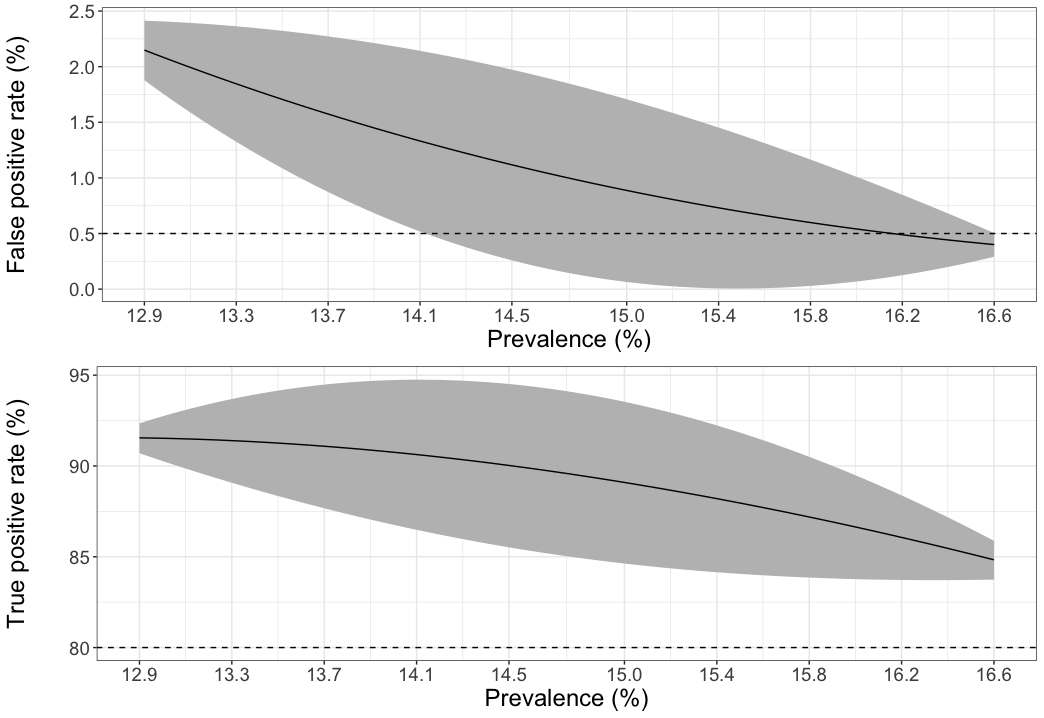}
\caption{Visualization of confidence set, $\wThetacalt$, for New York State study.
We see that this study gives strong evidence for prevalence in the range 13\%-16.6\%.
 }
\label{fig:NY}
\end{figure}

Finally, in Figure~\ref{fig:SC+LA+NY} of Appendix~\ref{appendix:serology} 
we present prevalence estimates for a combination of all datasets presented so far, 
while using both constructions, $\wThetac$ and $\wThetacalt$, to illustrate their differences.
As mentioned earlier, this requires the assumption that the antibody testing kits used 
in all three studies had identical specifications, or at least very similar so that the comparison remains informative. This assumption is most likely untenable given the available knowledge. However, we present the results there for illustration and completeness. 
The general picture in the combined study  is a juxtaposition of earlier findings. For example, 
both false and positive rates are now important for identification. The identified prevalence is 
in the range 5.2\%-8.2\% in $\wThetacalt$ (and 3.2\%-8.9\% in $\wThetac$).
These numbers are larger than the Santa Clara/LA county studies but smaller than the New York study.

\section{Discussion}\label{sec:discussion}

\subsection{Computation}\label{sec:discuss_computation}
The procedure described in Section~\ref{sec:proc} is computationally intensive for two main reasons. First, we need to consider all values of $\btheta\in\Theta$, which is a three-dimensional grid. Second, given some $\btheta$, we need to calculate $f(\bs | \btheta)$ for each $\bs\in\mathbb{S}$, which is also a three-dimensional grid. 

To deal with the first problem we can use parallelization, since the test decisions
in step 3 of our procedure are independent of each other. 
For instance, the results in Section~\ref{sec:application} were obtained in a computing cluster (managed by Slurm) comprised of 500 nodes, each with x86 architecture, 64-bit processors, and 16GB of memory. The total wall clock time to produce all results of the previous section was about 1 hour. The results for, say, the Santa Clara study can be obtained in much shorter time (a few minutes) because they contain few positive test results.
To address the second computational bottleneck we can  exploit the independence property between $\SCalNeg, \SCalPos$, 
and $\SMain$, as shown in the product of Equation~\eqref{eq:f}. 
Since any zero term in this product implies a zero value for $f$, we can ignore 
all individual term values that are very small. Through numerical experiments, we estimate that this computational trick prunes on average 97\% of $\St$ leading to a significant computational speedup. For example, 
to test one single value $\btheta\in\Theta$ takes about 0.25 seconds in a typical high-end laptop, which is a 200-fold speedup compared to 50 seconds required by the likelihood ratio test of Section~\ref{sec:LRT} --- see Appendix~\ref{appendix:LRT} for more details.

\subsection{Extrapolation to general population}\label{sec:bias}
As mentioned earlier, prevalence $\pi$ in Equation~\eqref{eq:pi} is a finite-population estimand, that is, it is a number that refers to the particular population in the study. 
Theorem~\ref{thm} shows that our procedure is valid for $\pi$ only under Assumption (A4).
However, to extrapolate to the general population we generally need to assume that
\begin{equation}\label{eq:A3}
\CalNeg, \CalPos, \Main~\text{are random samples from the population.}~\tag{A3}
\end{equation}
This is currently an untenable assumption. For example, in the Santa Clara study
the population of middle-aged white women was overrepresented, while the population of 
Asian or Latino communities was underrepresented. The impact from such selection bias on the inferential task is very hard to ascertain in the available studies.
Techniques such as post-stratification or reweighing can help, but at this early stage any extrapolation using distributional assumptions would be too speculative.
However, selection bias is a well-known issue among researchers, and can be addressed as 
widespread and carefully designed antibody testing catches on.
We leave this for future work.

\section{Concluding Remarks}
In this paper, we presented a partial identification method for estimating prevalence of Covid-19 from randomized serology studies. The benefit of our method is that it is valid in finite samples, as it does not rely on asymptotics, approximations or normality assumptions. We show that some recent serology studies in the US are not conclusive~(0\% prevalence is in the 95\% confidence set). 
However, the New York study gives strong evidence for high prevalence in the range
12.9\%-16.6\%. A combination of all datasets shifts this range down to 5.2\%-8.2\%, under a test uniformity assumption.
Looking ahead, we hope that the method developed here can contribute to a more robust analysis of future Covid-19 serology tests.

\section{Acknowledgments}
I would like to thank Guanglei Hong, Ali Hortascu, Chuck Manski, Casey Mulligan, 
Joerg Stoye, and Harald Uhlig for useful suggestions and feedback.
Special thanks to Connor Dowd for his suggestion of the alternative construction~\eqref{eq:Theta2}, 
and to Elie Tamer for various important suggestions. 
Finally, I gratefully acknowledge support from the John E. Jeuck Fellowship at Booth School of Business.

\small
\bibliographystyle{ims}
\bibliography{refs}

\newpage
\appendix

\normalsize
\section*{Appendix}

\newcommand{\logit}{\mathrm{logit}}
\newcommand{\eps}{\epsilon}
\newcommand{\bpsi}{\boldsymbol{\psi}}

\section{More details on the likelihood ratio test}\label{appendix:LRT}
The concrete testing procedure for the likelihood ratio test of Section~\ref{sec:LRT} is as follows.
\begin{enumerate}
\item Define the test statistic:
$$
t(\bs | \btheta) = \frac{f(\bs | \btheta)}{\max_{\btheta'\neq\btheta} f(\bs | \btheta')}.
$$
\item Calculate the observed value $t_\obs= t(\sobs | \btheta)$.
\item Sample $\{\bs^{(j)}, j=1, \ldots, r\}$ from $f(\bS | \btheta)$~(we set $r=1,000$ samples).
\item Calculate the one-sided p-value as
$
\frac{1}{r}\sum_{j=1}^r \mathbb{I}\{  t(\bs^{(j)} | \btheta) \ge t_\obs \}.
$
\end{enumerate}

For the maximization in the denominator of $t(\bs | \btheta)$ we use 
the standard BFGS algorithm on a natural re-parameterization. 
Specifically, we define the natural parameters as follows:
$$
\psi_0 = \logit(p); \psi_1 = \logit(q); \psi_2 = \logit(\pi),
$$
where $(p, q, \pi)\equiv \btheta$ are the original model parameters, i.e., 
false positive rate, true positive rate, and prevalence, respectively; 
and $\logit(z) = \log(z / (1-z)), z\in(0, 1)$. To avoid numerical instabilities 
we define $\logit(0) = \log(\eps / (1-\eps))$ and $\logit(1) = \log( (1-\eps) / \eps)$, where $\epsilon$ is a small constant, 
e.g., $\eps=$1e-8.
Since the natural parameter $ \bpsi = (\psi_0, \psi_1, \psi_2)$ is unconstrained, 
the optimization routine becomes faster and easier; 
mapping back from $\boldsymbol \psi$ to $\btheta$ is also straightforward.

The maximization takes about 0.05 seconds of wall-clock time in a typical high-end laptop.\footnote{
CPU: Intel(R) Core(TM) i7-8559U CPU at 2.70GHz; Memory: 16 GB 2133 MHz LPDDR3.}
It therefore takes a total of 50 seconds to test one single hypothesis based on  1,000 samples of the likelihood ratio. In contrast, our partial identification method takes 0.25 seconds of wall-clock time 
to test the same single null hypothesis, a 200-fold speedup. 
As explained in Section~\ref{sec:discuss_computation} this is because the 
computation of $f(\bs|\btheta)$ can be done very efficiently due to the 
decomposition of $f$ into three independent terms in Equation~\eqref{eq:f}.

Since the likelihood ratio test cannot be fully implemented, we chose to sample randomly 5,000 
parameter values from the basic confidence set, $\wThetac$, and 5,000 values from $\Theta \setminus\wThetac$, and then 
test each value using the likelihood ratio test. The idea is to explore the agreement of the two tests. The overlap between the likelihood ratio test decisions 
and the basic construction is 97.3\% for the values from $\wThetac$, and 97.7\% for the 
values from $\Theta \setminus\wThetac$. There was even more agreement with the alternative construction, specifically 97.3\% and 99.6\%, respectively.
Figure~\ref{fig:LRT} also shows some more detailed results. 
The x-axis represents the p-value generated in step 4 of the likelihood ratio test. 
The y-axis represents the ``strength of evidence" calculated by our procedure. 
For the basic construction, $\wThetac$, of Equation~\eqref{eq:Theta1} 
the strength of evidence is defined as $f(\sobs|\btheta) \nu(f(\sobs | \btheta), \btheta)$ --- 
the larger this value is, the stronger we reject. 
For the alternative construction, $\wThetacalt$, of Equation~\eqref{eq:Theta2} 
the value is defined as $\sum_{\bs\in\St} \mathbb{I}\{f(\bs | \btheta) \le f(\sobs|\btheta)\}
f(\bs | \btheta)$. We see high correlation between the different tests~(0.94 and 0.90, respectively).

\begin{figure}[h!]
\centering
\includegraphics[scale=0.45]{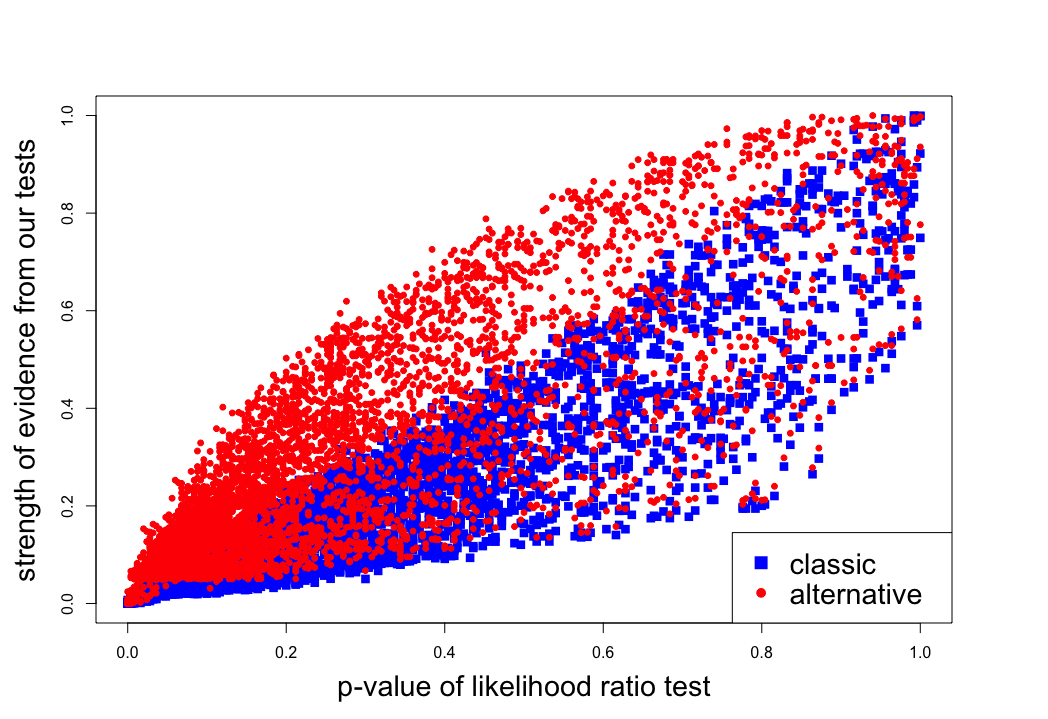}
\caption{Relationship between likelihood ratio test and our testing procedures, 
namely the ``classic" construction of Equation~\eqref{eq:Theta1}, and the 
alternative construction in Equation~\eqref{eq:Theta2}. }
\label{fig:LRT}
\end{figure}

\clearpage

\begin{figure}[h!]
\centering
Santa Clara study\\

\includegraphics[scale=0.36]{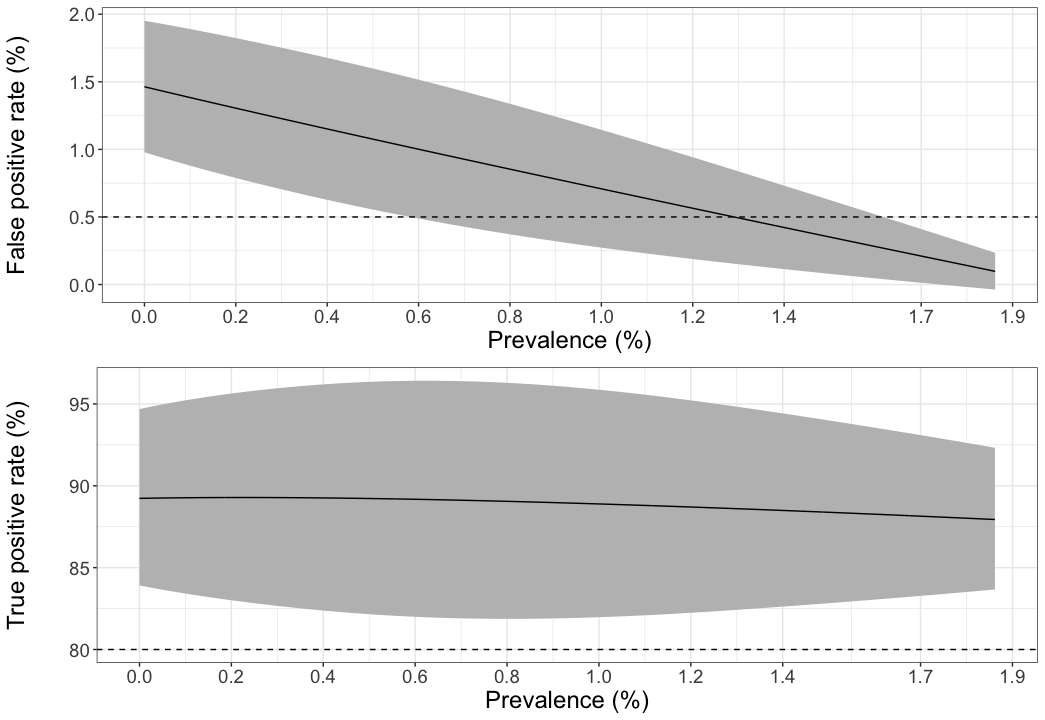}
\end{figure}

\begin{figure}[h!]
\centering
\vspace{5px}
LA county study \\

\includegraphics[scale=0.36]{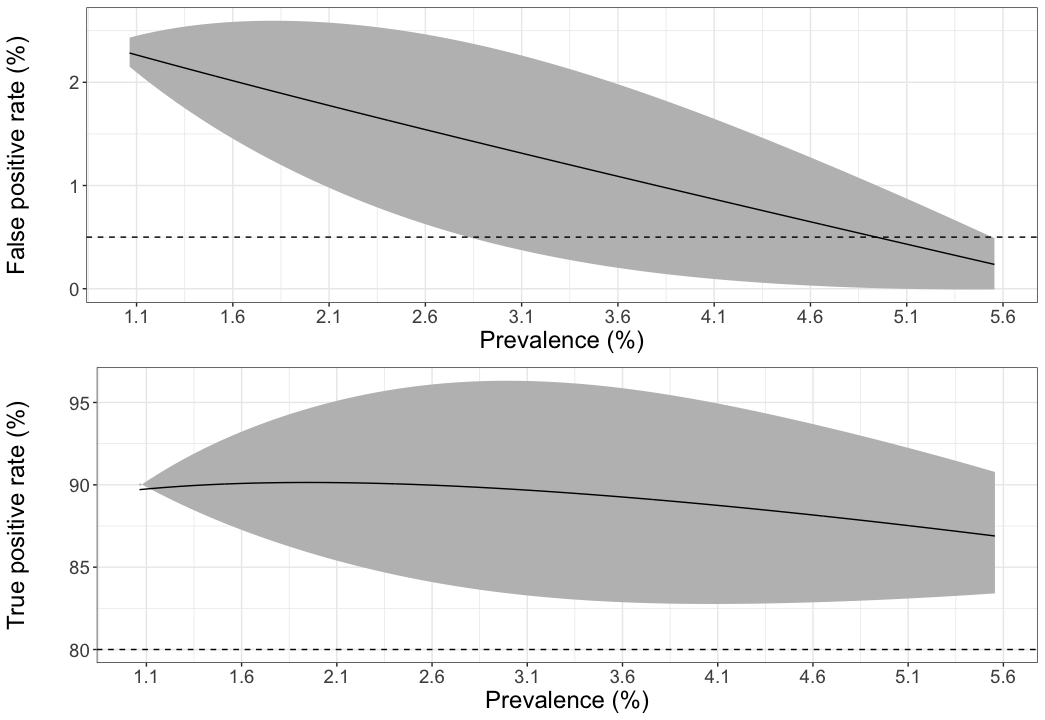}
\caption{95\% confidence set from likelihood ratio test-based procedure 
of Section~\ref{sec:LRT} for Santa Clara study (top) and LA county study (bottom).
}
\label{fig:LRT_1}
\end{figure}

\clearpage

\begin{figure}[h!]
\centering
\vspace{5px}
Santa Clara study and LA county study, combined \\

\includegraphics[scale=0.45]{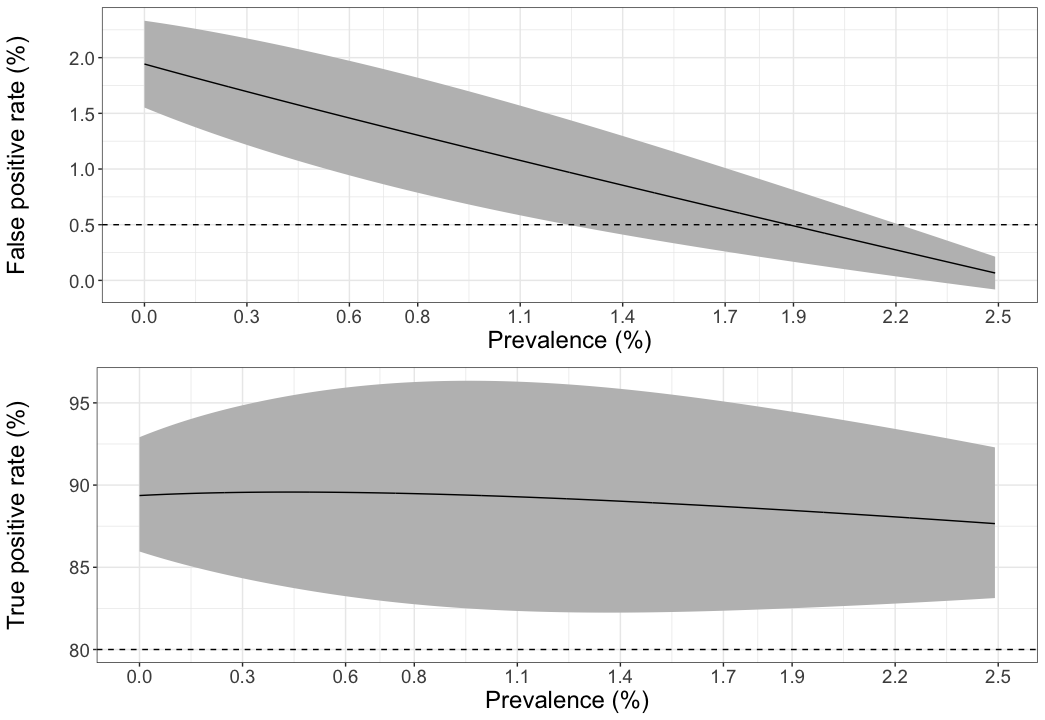}
\caption{95\% confidence set from likelihood ratio test-based procedure 
of Section~\ref{sec:LRT} for Santa Clara study and LA county study combined.
}
\label{fig:LRT_2}
\end{figure}

\clearpage

\section{More details on Monte Carlo confidence sets}\label{appendix:MCMC}
To implement the method of~\citet{chen2018monte} we used the following procedure.
\begin{enumerate}
\item We defined the natural parameter, $\boldsymbol{\psi}$ as in the previous section, 
and imposed a uniform prior: $f(\psi) \propto 1$.
\item We implement the Metropolis-Hastings algorithm through a symmetric proposal distribution, 
$q(\bpsi' | \bpsi) \sim N(\widehat\bpsi, \sigma^2 I)$, where 
$\widehat\bpsi$ is the maximum-likelihood estimate.
\item We run the Metropolis-Hastings for 200,000 iterations
and got samples from the posterior distribution 
$f(\bpsi | \sobs) \propto f(\sobs | \bpsi) f(\bpsi)$. We discarded the first 20\% of the posterior samples. Figure~\ref{fig:mcmc2} shows that the MCMC chain appears to be mixing well.
\item We transformed the $\bpsi$ samples back as $\btheta$ samples, and 
calculated $q_n$, the 95\% percentile of 
$\{ f(\sobs | \btheta^{(j)}),~j=1, \ldots\}$, where $\btheta^{(j)}$ denotes the $j$-th sample.
The 95\% confidence set is then defined as: 
$$
\wTheta = \{\btheta\in\Theta: f(\sobs | \btheta) \ge q_n\}.
$$

\end{enumerate}

 The 95\% posterior credible interval (shown in the bottom panel of Figure~\ref{fig:mcmc2}) was 
0.27\%-1.56\%, which is similar but slightly more narrow than the intervals from the Bayesian analyses 
described in Section~\ref{sec:compare2}.
The 95\% confidence set, $\wTheta$, defined above is visualized in Figure~\ref{fig:mcmc}.
Simple projection, yields a prevalence in the range 0.9\%-1.43\%. We note that 0.9\% corresponds to 32 true positives out of total 50 positives in the Santa Clara study. 
We can see from Figure~\ref{fig:mcmc2} that this number corresponds to the mode of the posterior marginal distribution for prevalence. Given the symmetry of this posterior distribution around the mode, 
it is surprising that the low-end of the confidence set corresponds to 0.9\% prevalence (i.e, 32 true positives out of 3,330 tests).
We can explain this discrepancy numerically by checking that values higher than 32~(prevalence higher than 0.9\%), that is, values on the right-end of the marginal posterior distribution, generally map to (much) higher likelihood values compared to values on the left-end. 
So, even though the marginal posterior looks symmetric, the likelihood values are not.
Since, the confidence set $\wTheta$ defined above is based on the likelihood values, 
it will be ``skewed" towards higher values of $\btheta$.

\begin{figure}[t!]
\centering
\includegraphics[scale=0.5]{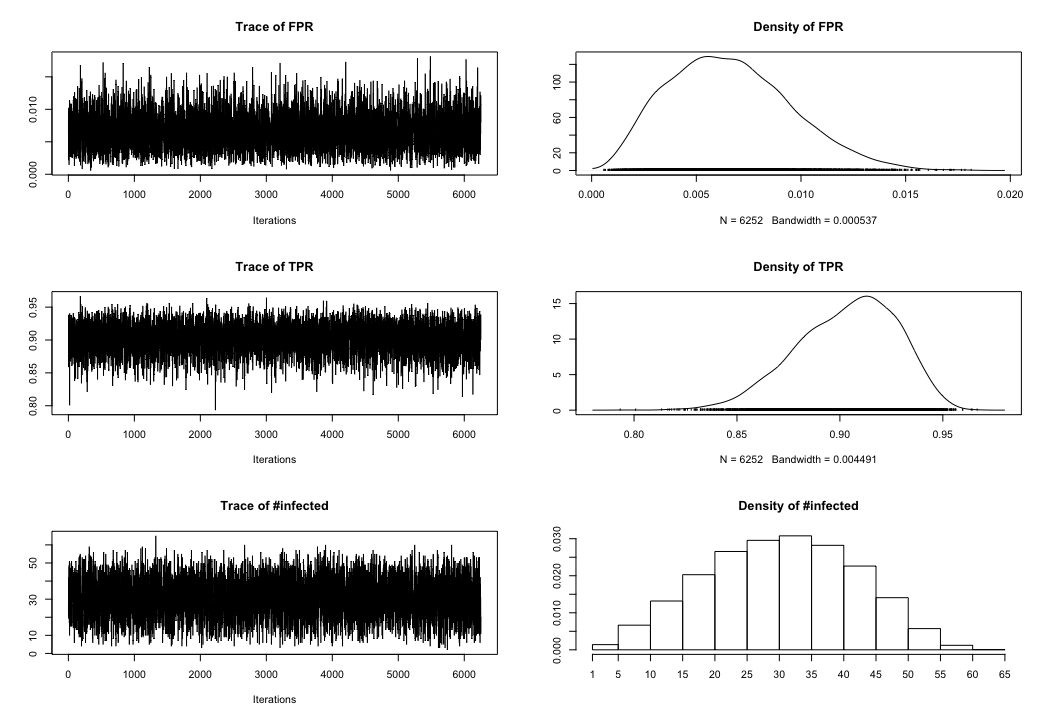}
\caption{Mixing plots for the MCMC chain used 
to construct the 95\% confidence set described in Section~\ref{sec:mcmc}.}
\label{fig:mcmc2}
\end{figure}

\clearpage 

\begin{figure}[t!]
\centering
\includegraphics[scale=0.45]{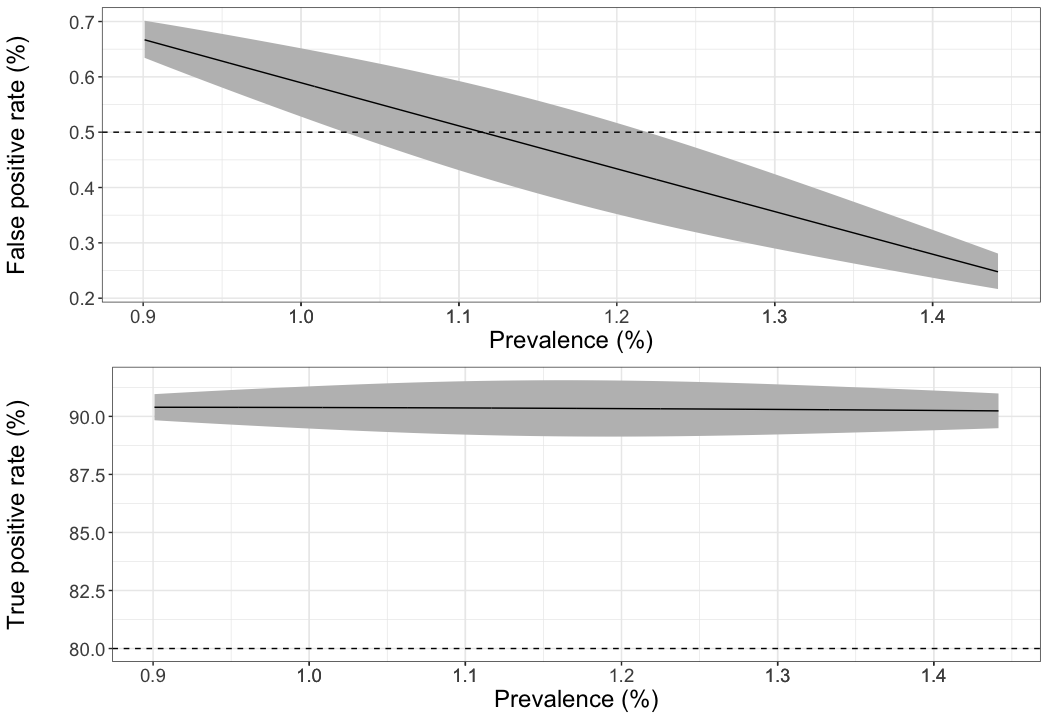}
\caption{The MCMC-based 95\% confidence set, $\wTheta$, produced by Procedure 1 of~\citet{chen2018monte}.}
\label{fig:mcmc}
\vspace{100px}
\end{figure}

\clearpage

\section{Additional results from serology studies in the US}
\label{appendix:serology}
Here, we present additional results from our analysis on serology studies from the US.
In particular, we combine datasets from various studies and analyze the results.
This requires the assumption that the antibody testing kits used 
in all three studies had identical specifications, or at least very similar so that the comparison remains informative. 

\begin{figure}[h!]
\vspace{40px}
\centering
\includegraphics[scale=0.38]{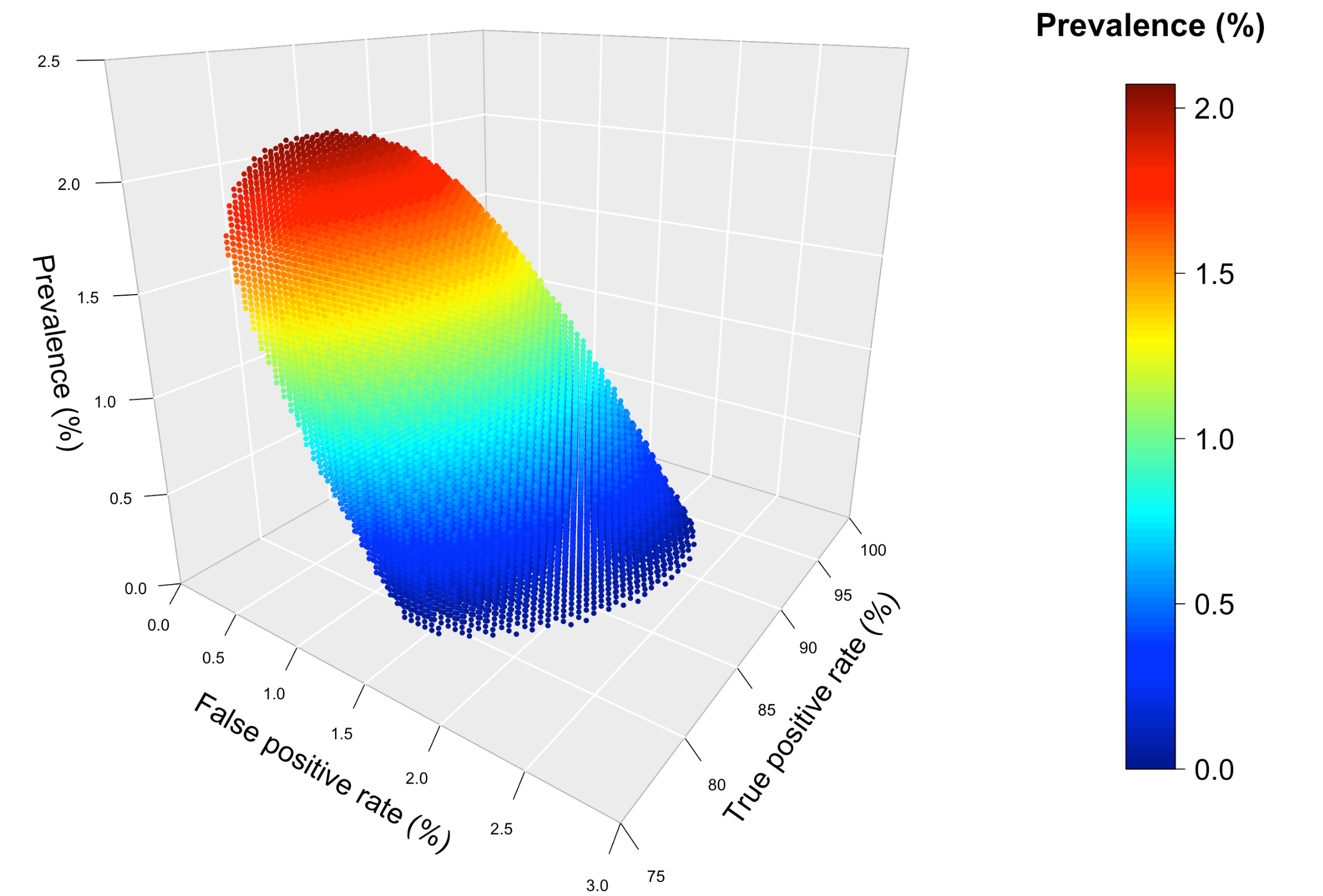}
\caption{Confidence set, $\wTheta_{0.95}$, for Santa Clara study, comprised of triples $\btheta=(p, q, \pi)$.
}
\label{fig:sc_Theta}
\end{figure}

\clearpage 

\begin{figure}[h!]
\centering
$\wThetac$ in Santa Clara \& LA county studies, combined
\includegraphics[scale=0.4]{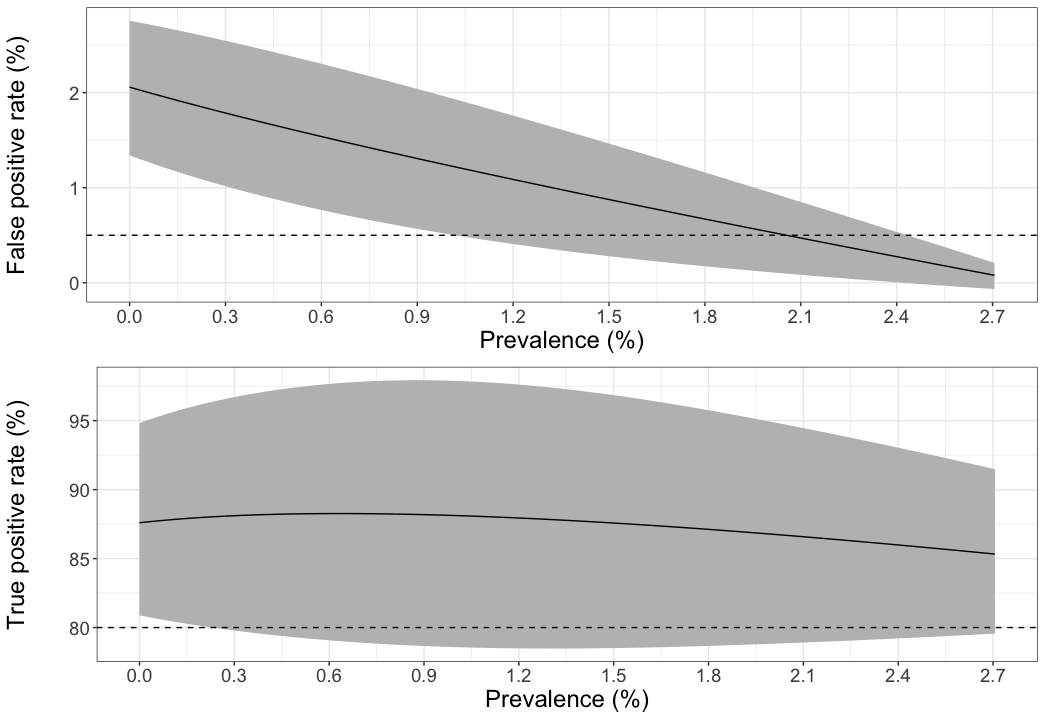}

\vspace{5px}
$\wThetacalt$ in Santa Clara \& LA county studies, combined
\includegraphics[scale=0.4]{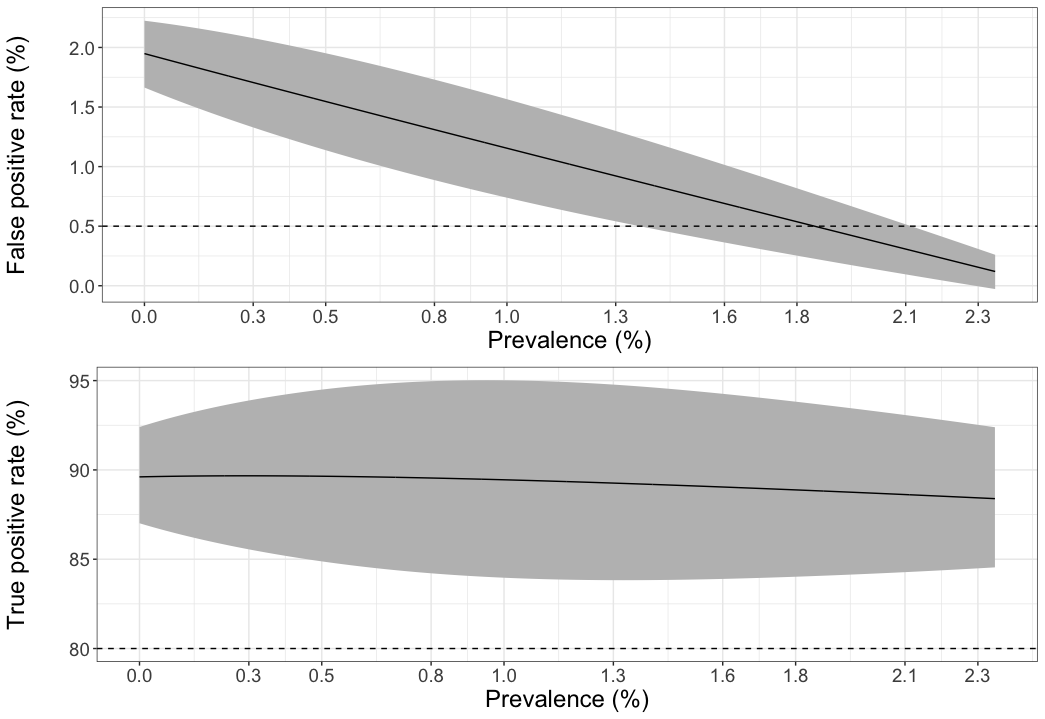}

\caption{Visualization of confidence sets, $\wThetac$ and $\wThetacalt$, for the Santa Clara and LA county studies, combined. The combined results are inconclusive (0\% prevalence is in the confidence sets). }
\label{fig:SC+LA}
\end{figure}

\newpage
\begin{figure}[t!]
\centering
$\wThetac$ in Santa Clara, LA county \& NY studies, combined
\includegraphics[scale=0.4]{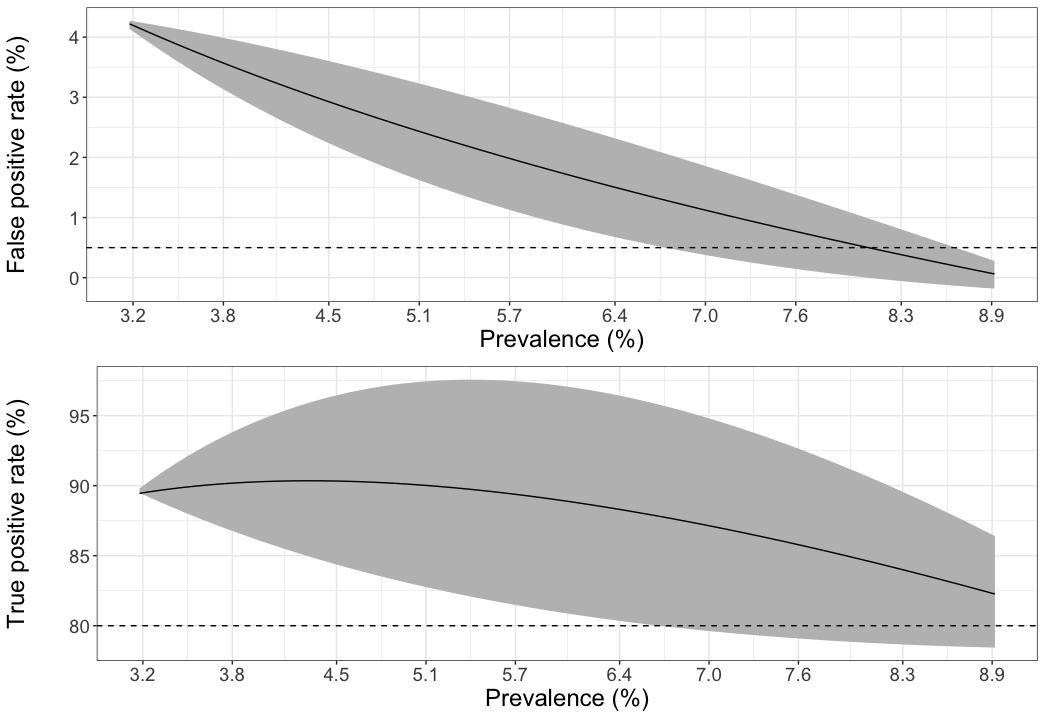}

$\wThetacalt$ in Santa Clara, LA county \& NY studies, combined

\vspace{5px}
\includegraphics[scale=0.4]{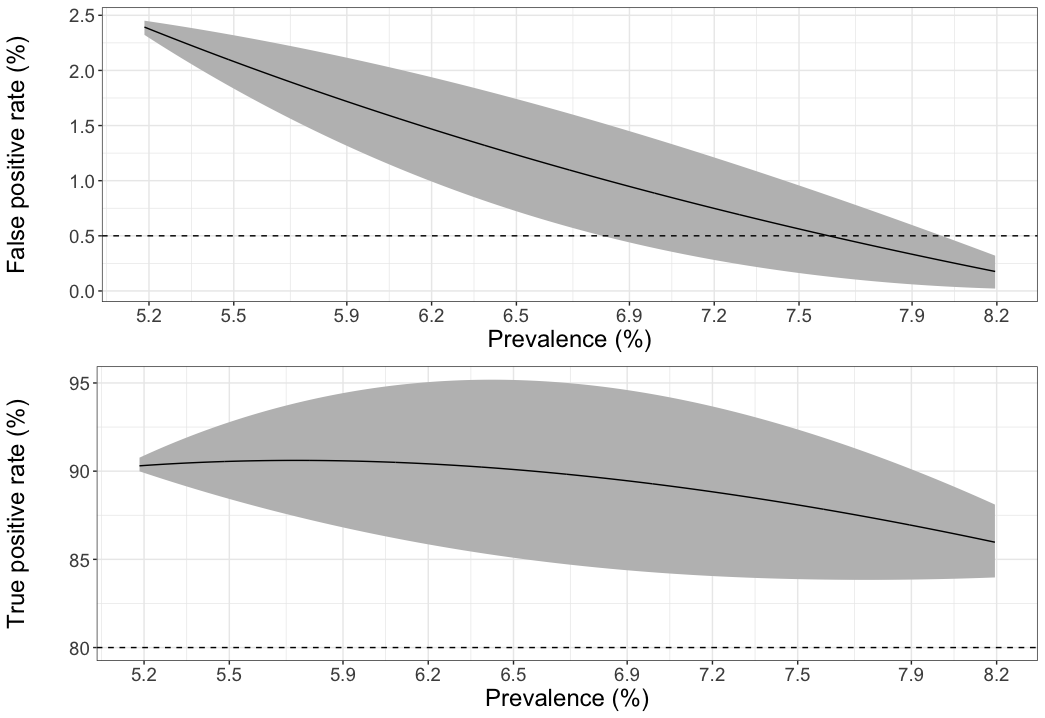}

\caption{Visualization of confidence sets, $\wThetac$ and $\wThetacalt$, for Santa Clara, LA county and New York studies, combined.
Assuming the data combination is valid, the combined dataset estimates Covid-19 prevalence in the range 5.2\%-8.2\%.
 }
\label{fig:SC+LA+NY}
\end{figure}

\clearpage

\section{Numerical example illustrating differences with likelihood-based methods}
\label{appendix:example}
Here, we illustrate the differences between our proposed inferential method 
and likelihood-based methods (both Bayesian and frequentist) through a simple numerical example. 
Suppose that $\St$ is such that $|\St| = N$, with $N$ extremely large, and fix some parameter value $\btheta_1$ to test.
Suppose also that the conditional density $f(\bS|\btheta_1)$ is defined as:
$$
f(\bs_0 | \btheta_1) = 0.95 - \epsilon_1, f(\bs_1 | \btheta_1) = \epsilon,~
\text{and}~
f(\bs | \btheta_1) = (0.05-\epsilon + \epsilon_1)/(N-2),~\forall\bs\in\St\setminus\{\bs_0, \bs_1\}.
$$
Set both $\epsilon$ and $\epsilon_1$ to be infinitesimal values.
As such, under $\btheta_1$ we observe $\bs_0$ with probability roughly equal to 0.95, 
or $\bs_1$ with some very small probability $\epsilon$, or 
observe any other remaining value from $\St$ uniformly at random.

Suppose we observe $\sobs = \bs_1$ in the data. Should we reject or accept $\btheta_1$? 

Since we can make $\epsilon$ arbitrarily small, an inferential method that focuses 
only on the likelihood function, would conclude that any $\btheta\in\Theta$ is more 
plausible than $\btheta_1$, as long as $f(\sobs | \btheta) >> \epsilon$.
Both frequentist and Bayesian methods would agree to such conclusion, 
and typically would perform inference around the mode of $f(\sobs | \btheta)$ with respect to $\btheta$.
However, our procedure makes a different conclusion, and actually 
accepts $\btheta_1$ (at the 5\% level)! The reason is that 
$$
\sum_{\bs\in\St} \mathbb{I}\{f(\bs | \btheta_1) \le f(\sobs | \btheta_1)\} f(\bs | \btheta_1)
= \epsilon + \frac{0.05-\epsilon + \epsilon_1}{N-2} (N-2) = 0.05 + \epsilon_1 > 0.05.
$$
That is, even though $f(\sobs |\btheta_1)$ is equal to a tiny value, 
there is still 5\% of the mass of $f(\bs|\btheta_1)$ at or below that value.

\end{document}